\documentclass[12pt,a4paper]{article}

\usepackage{amsfonts,amsmath,amsthm}
\usepackage{mathrsfs}
\usepackage{times}

\makeatletter

 \newcommand{\lyxaddress}[1]{
   \par {\raggedright #1 
   \vspace{1.4em}
   \noindent\par}
 }

\makeatother

\newtheorem{theorem}{Theorem}
\newtheorem{lemma}[theorem]{Lemma}
\newtheorem{proposition}[theorem]{Proposition}
\newtheorem{corollary}[theorem]{Corollary}
\newtheorem{remark}[theorem]{Remark}
\newtheorem*{prop*}{Lemma}
\theoremstyle{definition}
\newtheorem{definition}[theorem]{Definition}
\newtheorem*{rem*}{Remark}
\newtheorem*{rems*}{Remarks}

\newcommand{\N}{\mathbb{N}}

\newcommand{\R}{\mathbb{R}}
\newcommand{\XX}{\mathcal{X}}
\newcommand{\YY}{\mathcal{Y}}

\newcommand{\UU}{\mathcal{U}}
\newcommand{\Hs}{\mathscr{H}}
\newcommand{\Bs}{\mathscr{B}}

\newcommand{\dd}{\mathrm{d}}
\newcommand{\const}{\mathrm{const}}
\newcommand{\im}{\mathrm{i}}

\newcommand{\Dom}{\mathop\mathrm{Dom}\nolimits}
\newcommand{\Ran}{\mathop\mathrm{Ran}\nolimits}
\newcommand{\diag}{\mathop\mathrm{diag}\nolimits}
\newcommand{\offdiag}{\mathop\mathrm{offdiag}\nolimits}
\newcommand{\ad}{\mathop\mathrm{ad}\nolimits}
\newcommand{\dist}{\mathop\mathrm{dist}\nolimits}

\newcommand{\ceil}[1]{\lceil#1\rceil}

\newcounter{point}

\begin{document}

\title{On the energy growth of some periodically driven quantum
  systems with shrinking gaps in the spectrum}

\author{Pierre Duclos$^{1}$, Ondra Lev$^{2}$, Pavel
  \v{S}\v{t}ov\'\i\v{c}ek$^{2}$}

\maketitle

\lyxaddress{$^{1}$ Centre de Physique Th\'eorique de Marseille UMR
  6207 - Unit\'e Mixte de Recherche du CNRS et des Universit\'es
  Aix-Marseille I, Aix-Marseille II et de l' Universit\'e du Sud
  Toulon-Var - Laboratoire affili\'e \`a la FRUMAM}

\lyxaddress{$^{2}$Department of Mathematics, Faculty of Nuclear
  Science, Czech Technical University, Trojanova 13, 120 00 Prague,
  Czech Republic}

\begin{abstract}\noindent
  We consider quantum Hamiltonians of the form $H(t)=H+V(t)$ where the
  spectrum of $H$ is semibounded and discrete, and the eigenvalues
  behave as $E_n\sim{}n^\alpha$, with $0<\alpha<1$. In particular, the
  gaps between successive eigenvalues decay as $n^{\alpha-1}$. $V(t)$
  is supposed to be periodic, bounded, continuously differentiable in
  the strong sense and such that the matrix entries with respect to
  the spectral decomposition of $H$ obey the estimate
  $\|V(t)_{m,n}\|\leq\varepsilon\,|m-n|^{-p}\max\{m,n\}^{-2\gamma}$
  for $m\neq{}n$ where $\varepsilon>0$, $p\geq1$ and
  $\gamma=(1-\alpha)/2$. We show that the energy diffusion exponent
  can be arbitrarily small provided $p$ is sufficiently large and
  $\varepsilon$ is small enough. More precisely, for any initial
  condition $\Psi\in\Dom(H^{1/2})$, the diffusion of energy is bounded
  from above as $\langle{}H\rangle_\Psi(t)=O(t^\sigma)$ where
  $\sigma=
  \alpha/(2\ceil{p-1}\gamma-\frac{1}{2})$.
  As an application we consider the Hamiltonian
  $H(t)=|p|^\alpha+\varepsilon{}v(\theta,t)$ on $L^2(S^1,\dd\theta)$
  which was discussed earlier in the literature by Howland.
  
\end{abstract}

\section{Introduction}

One of the basic questions one can ask about time-dependent quantum
systems is the growth of energy on a long time scale for a given
initial condition. Unfortunately the quantum dynamics in the
time-dependent case proved itself to be rather difficult to analyze in
its full generality and complexity. The systems which allow for at
least partially analytical treatment and whose dynamics has been
perhaps best studied from various points of view are either driven
harmonic oscillators \cite{BJLPN,HLS,C1,EV} or periodically kicked
quantum Hamiltonians \cite{C2,C3,DeBF,Bou,Bou2,McCMcK}. On a more
general level, it is widely believed that there exist close links
between long time behavior of a quantum system and its spectral
properties. For time-independent quantum systems such a relation is
manifested by the famous RAGE theorem, see \cite{RS} for a summary and
references to the original papers. In a modified form this theorem has
been extended to periodic and quasi-periodic quantum systems
\cite{EV,JL,dO}. In this case the relevant operator whose spectral
properties are of interest is the Floquet (monodromy) operator.
Naturally, much attention has been paid to the spectral analysis of
Floquet operators in some of the papers cited above, see also
\cite{ABCF} for more recent results. Let us mention that a refined
analysis of how the spectral properties determine the quantum dynamics
is now available, see for example \cite{GM,Cb} and other papers, but
here we are not directly concerned with this question.

Thus for periodically time-dependent systems one can distinguish as a
related problem the spectral analysis of the Floquet operator under
certain assumptions on the quantum Hamiltonian. Frequently one writes
the time-dependent Hamiltonian in the form $H(t)=H+V(t)$ while
imposing assumptions on the spectral properties of the unperturbed
part $H$ and requiring some sort of regularity from the perturbation
$V(t)$. For our purposes an approach is rather important which is
based on the adiabatic methods and which was initiated by Howland
\cite{H1,H2} and further extended in \cite{N1,J1}. An essential
property imposed on the unperturbed Hamiltonian in this case is the
discreteness of the spectrum with increasing gaps between successive
eigenvalues.
 
Under this hypothesis Nenciu in \cite{N2} was not only able to
strengthen the results due to Howland but he derived in addition an
upper bound on the diffusive growth of the energy having the form
$\const\,t^{a/n}$ where $a>0$ is given by the spectral properties of
$H$ and $n$ is the order of differentiability of $V(t)$.  Inspired by
this result on the energy growth, Joye in \cite{J2} considered another
class of time-dependent quantum Hamiltonians with rather mild
assumptions on the spectral properties of $H$ but, on the other hand,
assuming that the strength of the perturbation $V(t)$ is in some sense
small with respect to $H$. Moreover, as far as the energy diffusion is
discussed, the periodicity of $V(t)$ is required neither in \cite{N2}
nor in \cite{J2}.

It is worthwhile to mention that Howland in \cite{H3} succeeded to
treat also the case when the spectrum of $H$ is discrete but the gaps
between successive eigenvalues are decreasing. To achieve this goal he
restricted himself to certain classes of perturbations $V(t)$
characterized by the behavior of matrix entries with respect to the
eigen-basis of $H$. In particular, he discussed as an example the
following model: $H(t)=|p|^\alpha+v(\theta,t)$ in
$L^2(S^1,\dd{}\theta)$ where $0<\alpha<1$ and $v(\theta,t)$ is in
$C^\infty(S^1\times{}S^1)$. It seems to be natural to look in this
case, too, for a result parallel to that due to Nenciu \cite{N2} and
to attempt a derivation of a nontrivial bound on the diffusive growth
of energy.  But we are aware of only one contribution in this
direction made by Barbaroux and Joye \cite{BJ}; it is based on the
general scheme proposed in \cite{J2}.

In this paper we wish to complete or to strengthen the results from
\cite{BJ} while making use of some ideas from \cite{J2}. Thus we aim
to consider other classes of time-dependent Hamiltonians whose
unperturbed part $H$ has a discrete spectrum with decreasing gaps. In
particular, the derived results are applicable to the Howland's model
introduced in \cite{H3}. In more detail, we deal with a quantum system
described by the Hamiltonian $H(t):=H+V(t)$ acting on a separable
Hilbert space $\Hs$ and such that $H$ is semibounded and has a pure
point spectrum with the spectral decomposition
\begin{displaymath}
  H = \sum_{n\in\N} E_nP_n.
\end{displaymath}
Assume that the eigen-values $E_1<E_2<\dots$ obey the shrinking gap
condition
\begin{equation}
  \label{VGC}
  c_H\frac{|m-n|}{\max\{m,n\}^{2\gamma}}
  \leq |E_m-E_n| \leq C_H\frac{|m-n|}{\max\{m,n\}^{2\gamma}}
\end{equation}
for some $\gamma\in\,]0,\frac12[$ and strictly positive constants
$c_H$, $C_H$.  Notice that condition (\ref{VGC}) implies
$E_n\sim{}n^\alpha$ where $\alpha=1-2\gamma\in\,]0,1[\,$ (more
precisely, (\ref{VGC}) implies that the sequence $E_nn^{-\alpha}$ is
bounded both from below and from above by strictly positive constants
for all sufficiently large $n$). To simplify the discussion let us
assume, without loss of generality, that $H$ is strictly positive,
i.e., $E_1>0$.

The time-dependent perturbation $V(t)\in\Bs(\Hs)$ is supposed to be
$T$-periodic and $C^1$ in the strong sense. From the strong
differentiability it follows that the propagator $U(t,s)$ associated
to the Hamiltonian $H+V(t)$ exists and preserves the domain $\Dom(H)$
(see, e.g., \cite{K}).

Let us suppose that $V$ is small with respect to the norm
\begin{equation}
  \label{eq:normVpgamma}
  \Vert V\Vert_{p,\gamma} := \sup_{t\in [0,T]} \sup_{m,n\in\N}
  \langle m-n\rangle^p \max\{m,n\}^{2\gamma}\,\Vert V(t)_{m,n}\Vert,
\end{equation}
where $p>2$,
\begin{displaymath}
  \langle m-n\rangle := \max\{1,|m-n|\},
\end{displaymath}
and $\Vert{}V(t)_{m,n}\Vert$ denotes the norm of the operator
\begin{displaymath}
  V(t)_{m,n} := P_m V(t) P_n :\Ran P_n\rightarrow \Ran P_m.
\end{displaymath}
We claim that if, in addition, $\ceil{p-1}>1/{(2(1-\alpha))}$ then the
propagator $U(t,s)$ preserves the form domain $Q_H=\Dom(H^{1/2})$ and
for any $\Psi$ from $Q_H$ one can estimate the long-time behavior of
the energy expectation value by
\begin{equation}
  \label{eq:intro_sigma}
  \langle U(t,0)\Psi,HU(t,0)\Psi\rangle = O(t^\sigma),
  \textrm{~with~~}
  \sigma = \frac{2\alpha}{2\ceil{p-1}(1-\alpha)-1}
\end{equation}
(more details are given in Theorem~\ref{dyn_stab} below). Here
$\ceil{x}$ is standing for the ceiling of a real number $x$, i.e., the
smallest integer greater than or equal to $x$.

Provided that $[V(t),V(s)]=0$ for every $t,s$ and
$\int_0^TV(t)\,\dd{}t=0,$ the assumption
$\Vert{}V\Vert_{p,\gamma}\le\varepsilon$ can be replaced by
$\Vert{}V\Vert_{p+1,0}\le\varepsilon$, i.e.,
\begin{displaymath}
  \Vert P_m V(t)P_n\Vert \le
  \frac{\varepsilon}{\langle m-n\rangle^{p+1}}.
\end{displaymath}
The condition $[V(t),V(s)]=0$ is satisfied for example when $V(t)$ is
a potential (i.e., a multiplication operator by a function on a
certain $L^2$ space) or when the time dependence of $V(t)$ is
factorized, i.e., $V(t)=f(t)v$ where $f(t)$ is a real-valued
($T$-periodic and $C^1$) function and $v$ is a time-independent
operator on $\Hs$.

Let us stress that even though the energy diffusion exponent $\sigma$
in (\ref{eq:intro_sigma}) can be made arbitrarily small provided $p$
is sufficiently large our result is still far away from the situation
when one can prove the dynamical stability in the sense that the
energy remains bounded in time for any initial condition \cite{DeBF}.
The point is that the time-dependent perturbation $V(t)$ is supposed
to be sufficiently regular and small by requiring that
$\|V\|_{p,\gamma}<\varepsilon$ where not only the norm but also the
positive bound $\varepsilon$ depends on $p$ (see
Theorem~\ref{dyn_stab} below for a precise formulation). This plays a
role also in the analysis of the Howland's model in
Subsection~\ref{sec:howlands_model}. In this case,
$H=|p|^\alpha+\varepsilon{}v(\theta,t)$ and the exponent $\sigma$ in
(\ref{eq:howland_sigma}) tends to $0$ as the order of
differentiability of $v(\theta,t)$ in $\theta$, called $k$, tends to
infinity. However the coupling constant $\varepsilon$ is supposed to
be sufficiently small in dependence on $k$ and so one cannot claim
that $\sigma$ equals $0$ even if $v(\theta,t)$ is smooth in $\theta$.

On the other hand, to our knowledge, non-trivial examples of
time-dependent quantum models for which one can verify this strong
type of dynamical stability are rather rare. A periodically
time-dependent quadratic Hamiltonian represents such a model. It is
explicitly solvable and this is how one can verify the boundedness of
energy in the non-resonant case \cite{EV}. A broader class of
periodically time-dependent models is shown to be dynamically stable
for non-resonant values of frequencies with the aid of the KAM
(Kolmogorov-Arnold-Moser) type method in \cite{ADE}, see also
\cite{DSSV} for some additional discussion. In this connection let us
point out a recent example \cite{dOS} showing that the relationship
between the spectral properties of the Floquet operator and the
dynamical stability is not so transparent, and it may require a
considerable amount of efforts to understand it properly.

Let us compare the result of the current paper, as briefly described
above, to the results derived in \cite{J2} and \cite{BJ}. Paper
\cite{J2} focuses on the general scheme and is not so much concerned
with particular cases as that one we are going to deal with here.
Nevertheless a possible application to the Howland's classes of
perturbations is shortly discussed in Proposition~5.1 and Lemma~5.1.
The Howland's classes are determined by a norm which somewhat differs
from (\ref{eq:normVpgamma}), as explained in more detail in
Subsection~\ref{sec:gap-cond-modif}. But the difference is not so
essential to prevent a comparison. To simplify the discussion let us
assume that the eigenvalues of $H$ are simple and behave
asymptotically as $E_n\sim\const\,n^\alpha$, with $0<\alpha<1$. In the
particular case when $\|V\|_{p,\gamma}<\infty$ for some $p>1$ and
$\gamma=(1-\alpha)/2$ the bound on the energy diffusion exponent
derived in \cite{J2} equals $\alpha/(2\gamma-\frac{1}{2})$ provided
$\gamma>(1+\alpha)/4$, i.e., $\alpha<1/3$. Our bound
$\alpha/(2\ceil{p-1}\gamma-\frac{1}{2})$, valid for $0<\alpha<1$ and
provided $p>2$ and $\ceil{p-1}>1/(4\gamma)$, is achieved by making use
of the rapid decay of matrix entries of $V$ in the direction
perpendicular to the diagonal. It follows that we can make the growth
of the energy $\langle{}H\rangle_\Psi$ arbitrarily slow by imposing
more restrictive assumptions on the perturbation $V$, i.e., by letting
the parameter $p$ be sufficiently large.

In paper \cite{BJ} one treats in fact a larger class of perturbations
than we do since one requires only the finiteness of the norm
$\|V\|_{p,0}<\infty$ for $p$ sufficiently large. In other words, no
decay of matrix entries of $V$ along the diagonal is supposed. On the
other hand, one assumes that the initial quantum state belongs to the
domain $\Dom(H^\beta)$ for $\beta$ sufficiently large; $\beta$ is
never assumed therein to be smaller than $3/2$.  Furthermore, there is
no assumption on the periodicity of $H(t)$ both in \cite{BJ} and
\cite{J2}. On the other hand, our assertion concerns all initial
states from the domain $\Dom(H^{1/2})$ but we need a decay of matrix
entries of $V$ along the diagonal at least of order
$2\gamma=1-\alpha$. For the sake of comparison let us also recall the
bound on the energy diffusion exponent which has been derived in
\cite{BJ}. It is roughly of the form $\alpha/(1-f(p))^2$ where
$\alpha$ has the same meaning as above, $f(p)$ is positive and
$f(p)=O(p^{-1})$ as $p\to\infty$. Hence this bound is never smaller
than $\alpha$ and approaches this value as the parameter $p$ tends to
infinity.

\section{Upper bound on the energy growth}

\subsection{The gap condition and the modified Howland's classes}
\label{sec:gap-cond-modif}

On the contrary to Howland who introduced in \cite{H3} the classes
$\XX(p,\delta)$ equipped with the norm
\begin{displaymath}
  \Vert A\Vert_{p,\delta}^H = \sup_{m,n}\,\{(mn)^\delta
  \langle m-n\rangle^p\,\Vert A_{m,n}\Vert;\textrm{~}m,n\geq1\}, 
\end{displaymath}
we prefer to work with somewhat modified classes, called
$\YY(p,\delta)$, whose definition is adjusted to the gap condition
(\ref{VGC}). Our choice is dictated by an expected asymptotic behavior
of eigenvalues of $H$ in a typical situation. Let us briefly explain
where condition (\ref{VGC}) comes from.

We expect the eigenvalues to behave asymptotically as
$E_n=\const\,n^\alpha(1+o(1))$ where the error term $o(1)$ is supposed
to tend to zero sufficiently fast. The spectral gaps $E_{n+1}-E_n$
tend to zero as $n\to\infty$ if $\alpha\in\,]0,1[$.  Keeping the
notation $\gamma:=(1-\alpha)/2$ we wish to estimate the difference
$|E_m-E_n|$. To this end we replace $E_n$ simply by the power sequence
$n^\alpha$. Then one gets
\begin{displaymath}
  \frac{m^\alpha-n^\alpha}{m-n}(m n)^\gamma
  = \frac{\sinh (\alpha y)}{\sinh(y)}
  = e^{-(1-\alpha)|y|} \frac{1-e^{-2\alpha|y|}}{1-e^{-2|y|}}
\end{displaymath}
where $e^{2y}:=m/n$. Since the fraction
$(1-e^{-2\alpha|y|})/(1-e^{-2|y|})$ can be estimated by positive
constants both from above and from below we finally find that
\begin{displaymath}
  C_1\,\frac{|m-n|}{\max\{m,n\}^{2\gamma} }\le |m^\alpha-n^\alpha|
  \le C_2\,\frac{|m-n|}{\max\{m,n\}^{2\gamma} }
\end{displaymath}
for some $C_1,C_2>0$ and all $m,n\in\N$.

\begin{definition}
  \label{Y(p,a)}
  Let $p\geq1$, $\delta\geq0$ and $p+2\delta>1$. We say that an
  operator $A\in\Bs(\Hs)$ belongs to the class $\YY(p,\delta)$ if and
  only if
  \begin{equation}
    \label{eq:def_norm_Y}
    \Vert A\Vert_{p,\delta} := \sup_{m,n\in\N}\,
    \langle m-n\rangle^p \max\{m,n\}^{2\delta}\,\Vert A_{m, n}\Vert
    < \infty.
  \end{equation}
  
  Let $A(t)$ be a $T$-periodic function with values in the space
  $\YY(p,\delta)$. With some abuse of notation we shall also write
  \begin{displaymath}
    \Vert A\Vert_{p,\delta} := \sup_{t\in[0,T]}\sup_{m,n\in\N}\,
    \langle m-n\rangle^p \max\{m,n\}^{2\delta}\,
    \Vert A(t)_{m, n}\Vert.
  \end{displaymath}
\end{definition}

\begin{rems*}
  \renewcommand{\thepoint}{\roman{point}}
  \setcounter{point}{0}
  \addtocounter{point}{1}
  (\thepoint)\, It is straightforward to check that
  $\|\cdot\|_{p,\delta}$ is indeed a norm. Let us note that an
  equivalent norm is obtained if one replaces $\max\{m,n\}$ by $(m+n)$
  in (\ref{eq:def_norm_Y}).
  \addtocounter{point}{1}

  (\thepoint)\, Obviously, $\YY(p,\delta)\subset\XX(p,\delta)$. Notice
  that $\YY(p,\delta)$ is a Banach space equipped with the norm
  $\Vert\cdot\Vert_{p,\delta}$.
  \addtocounter{point}{1}

  (\thepoint)\, For the sake of convenience we have chosen the norm
  (\ref{eq:def_norm_Y}) with the restrictions $p\geq1$, $\delta\geq0$
  and $p+2\delta>1$ so that if it is finite for a matrix $\{A_{mn}\}$,
  $A_{mn}\in\Bs(\Ran{}P_n,\Ran{}P_m)$, then the matrix corresponds to
  a bounded operator $A\in\Bs(\Hs)$. Indeed, it is so since one can
  estimate the operator norm $\|A\|$ by the Shur-Holmgren norm
  \begin{displaymath}
    \Vert A\Vert_{S H} := \max\!
    \left\{\sup_{m\in\N}\sum_{n\in\N}\Vert A_{m,n}\Vert,
      \sup_{n\in\N} \sum_{m\in\N}\Vert A_{m,n}\Vert\right\}.
  \end{displaymath}
  It clearly holds
  \begin{displaymath}
    \|A\|_{SH} \leq \|A\|_{p,\delta}\sup_{m\in\N}\,\sum_{n=1}^\infty
    \frac{1}{\langle m-n\rangle^p\max\{m,n\}^{2\delta}}\,.
  \end{displaymath}
  The sum on the RHS equals
  \begin{eqnarray*}
    \frac{1}{m^{2\delta}}+\sum_{n=1}^{m-1}\frac{1}{(m-n)^pm^{2\delta}}
    +\sum_{n=m+1}^{\infty}\frac{1}{(n-m)^pn^{2\delta}}
    &\leq& 2+\frac{1}{m^{2\delta}}\int_1^m\frac{\dd x}{x^p}
    +\sum_{k=1}^{\infty}\frac{1}{k^{p+2\delta}} \\
    &=& 2+\frac{1-m^{-p+1}}{(p-1)m^{2\delta}}+\zeta(p+2\delta).
  \end{eqnarray*}
  Setting temporarily $x=\ln(m)$ and $\epsilon=p-1$ one can make use
  of the inequality
  \begin{displaymath}
    \frac{1}{\epsilon}
    \left(e^{-2\delta x}-e^{-(\epsilon+2\delta)x}\right)
    \leq \frac{1}{\epsilon+2\delta}
  \end{displaymath}
  which is true for all $x\geq0$ provided $\epsilon\geq0$,
  $\delta\geq0$ and $\epsilon+2\delta>0$. Thus one arrives at the
  estimate
  \begin{displaymath}
    \|A\|_{SH}
    \leq \left(2+\frac{1}{p+2\delta-1}+\zeta(p+2\delta)\right)
    \|A\|_{p,\delta}.
  \end{displaymath}
  Here $\zeta(u):=\sum_{k=1}^\infty{}k^{-u}$ denotes the Riemann's
  zeta function.
  \addtocounter{point}{1}

  (\thepoint)\, Finally let us note that the value $p=\infty$ is
  admissible. We shall use the norm $\|\cdot\|_{\infty,\delta}$
  exclusively in the case of diagonal matrices when it simply reduces
  to
  \begin{displaymath}
    \Vert A\Vert_{\infty,\delta} := \sup_{n\in\N}\,
    n^{2\delta}\,\Vert A_{n, n}\Vert.
  \end{displaymath}
\end{rems*}

From Definition~\ref{Y(p,a)} one immediately deduces the following
lemma.

\begin{lemma}
  \label{how}
  Suppose that $H$ is an operator on $\Hs$ with pure point spectrum
  whose eigen-values $E_1<E_2<\dots$ obey the upper bound in
  (\ref{VGC}). Let $p>2$. If $A\in\YY(p,\delta)$ then the commutator
  $[A,H]$ lies in $\YY(p-1,\delta+\gamma)$ and
  \begin{displaymath}
    \Vert[A,H]\Vert_{p-1,\delta+\gamma} \le C_H \Vert A\Vert_{p,\delta}.
  \end{displaymath}
\end{lemma}

A basic technical tool we need is a lemma concerned with products of
two classes $\YY$. For its proof as well as for the remainder of the
paper the following two elementary inequalities will be useful.
According to the first one, for every $m,k\ge1$ it holds
\begin{equation}
  \label{<} 
  \frac m k\le 2\langle m-k\rangle.
\end{equation}
In fact, this is a direct consequence of the implication
$a,b\ge1\implies{}a+b\le2ab$.

The second inequality claims that if $a,b\geq0$ then
\begin{displaymath}
  \frac{\langle a+b\rangle}{\langle a\rangle\langle b\rangle}
  \leq \frac{2}{\langle\min\{a,b\}\rangle}\,.
\end{displaymath}
This can be reduced to the inequality
$\langle2a\rangle\leq2\langle{}a\rangle$ which is quite obvious.

\begin{lemma}
  \label{product_YY}
  Consider two classes $\YY(p_1,\delta_1)$, $\YY(p_2,\delta_2)$, with
  $p_1,p_2>1$, $\delta_1,\delta_2\geq0$. Suppose that numbers $p$,
  $\delta$ satisfy the inequalities
  \begin{displaymath}
    1<p\leq\min\{p_1,p_2\},\textrm{~}
    \max\{\delta_1,\delta_2\}\leq\delta\leq\delta_1+\delta_2,\textrm{~}
    p+2\delta\leq\min\{p_1+2\delta_1,p_2+2\delta_2\}.
  \end{displaymath}
  If $A\in\YY(p_1,\delta_1)$ and $B\in\YY(p_2,\delta_2)$ then
  \begin{equation}
    \label{eq:AB_pdelta}
    \|AB\|_{p,\delta} \leq C(p,\delta-\delta_0)\,
    \|A\|_{p_1,\delta_1}\|B\|_{p_2,\delta_2}
  \end{equation}
  where
  \begin{displaymath}
    C(p,\Delta) = 2^{p+2\Delta+1}(1+2\zeta(p))
  \end{displaymath}
  and $\delta_0=\min\{\delta_1,\delta_2\}$. Consequently,
  $\YY(p_1,\delta_1)\YY(p_2,\delta_2)\subset\YY(p,\delta)$.
\end{lemma}

\begin{proof}
  Under the assumptions we have
  \begin{displaymath}
    \langle m-n\rangle^p\max\{m,n\}^{2\delta}\|(AB)_{mn}\|
    \leq \langle m-n\rangle^p\max\{m,n\}^{2\delta}
    \sum_{\ell=1}^\infty\|A_{m\ell}\|\|B_{\ell n}\|
  \end{displaymath}
  which is less than or equal to
  \begin{equation}
    \label{eq:ABsum}
    \|A\|_{p_1,\delta_1}\|B\|_{p_2,\delta_2}
    \sum_{\ell=1}^\infty
    \frac{\langle m-n\rangle^{p}\max\{m,n\}^{2\delta}} 
    {\langle m-\ell\rangle^{p_1}\max\{m,\ell\}^{2\delta_1} 
      \langle n-\ell\rangle^{p_2}\max\{n,\ell\}^{2\delta_2}}\,.
  \end{equation}
  The summand in (\ref{eq:ABsum}) can be estimated from above by
  \begin{displaymath}
    \frac{\langle|m-\ell|+|n-\ell|\rangle^{p}\max\{m,n\}^{2\delta}}
    {\langle m-\ell\rangle^{p_1}\max\{m,\ell\}^{2\delta_1}
      \langle n-\ell\rangle^{p_2}\max\{n,\ell\}^{2\delta_2}}
    \leq\frac{2^p}{\langle\min\{|m-\ell|,|n-\ell|\}\rangle^p}\,
    h(m,n,\ell)
  \end{displaymath}
  where
  \begin{displaymath}
    h(m,n,\ell) = \langle m-\ell\rangle^{p-p_1}
    \langle n-\ell\rangle^{p-p_2}\,
  \frac{\max\{m,n\}^{2\delta}}{\max\{m,\ell\}^{2\delta_1}
    \max\{n,\ell\}^{2\delta_2}}\,.
  \end{displaymath}
  One can further estimate $h(m,n,\ell)$. For definiteness let us
  suppose that $m\geq{}n$. Then, since $p-p_2\leq0$ and
  $\delta_2\geq\delta-\delta_1$,
  \begin{eqnarray*}
    h(m,n,\ell) &\leq& \langle m-\ell\rangle^{p-p_1}
    \frac{m^{2\delta}}
    {\max\{m,\ell\}^{2\delta_1}\max\{n,\ell\}^{2(\delta-\delta_1)}}
    \,\leq\, \langle m-\ell\rangle^{p-p_1}
    \left(\frac{m}{\ell}\right)^{\!2(\delta-\delta_1)}\\
    &\leq& 2^{2(\delta-\delta_1)}
    \langle m-\ell\rangle^{p-p_1+2(\delta-\delta_1)}
    \,\leq\, 2^{2(\delta-\delta_0)}.
  \end{eqnarray*}
  It follows easily that the sum in (\ref{eq:ABsum}) is bounded from
  above by $2^{p+2(\delta-\delta_0)}(2+4\zeta(p))$ and this estimate
  implies (\ref{eq:AB_pdelta}).
\end{proof}

\begin{corollary}
 \label{lemma_XY}
 Let $p>2$, $i\ge1$ and $\gamma\in\,]0,\frac12[$. Then the following
 product formulas hold true:
 \begin{eqnarray*}
   &&\YY(p,i\gamma)\,\YY(p,i\gamma) \subset \YY(p-1,(i+1)\gamma) \\
   &&\YY(p,(i-1)\gamma)\,\YY(p-1,i\gamma) \subset \YY(p-1,i\gamma) \\
   &&\YY(p+1,(i-1)\gamma) \,\YY(p-1,(i+1)\gamma)
   \subset \YY(p-1,(i+1)\gamma)
 \end{eqnarray*}
 The formulas are also true for the opposite order of factors on the
 LHS. Moreover, if operators $A$ and $B$ belong to the corresponding
 classes on the LHS then
 \begin{eqnarray*}
   &&\Vert A B\Vert_{p-1,
     (i+1)\gamma} \leq C_{p}\,\Vert A\Vert_{p, i\gamma}
   \Vert B\Vert_{p, i\gamma}\\
   &&\Vert A B\Vert_{p-1, i \gamma}
   \leq C_{p}\,\Vert A\Vert_{p,(i-1)\gamma}
   \Vert B\Vert_{p-1, i\gamma}\\
   &&\Vert A B\Vert_{p-1, (i+1)\gamma} \leq 2C_{p}\,
   \Vert A\Vert_{p+1,(i-1)\gamma}\Vert B\Vert_{p-1,(i+1)\gamma},
 \end{eqnarray*}
 where
 \begin{equation}
   \label{norm_const}
   C_{p}:= 2^{p+1}(1+2\,\zeta(p-1)).
 \end{equation}
 The norm estimates hold true also for the opposite order of factors
 $A$ and $B$ in the product.
\end{corollary}

\subsection{The main theorem}

\begin{theorem}
  \label{dyn_stab}
  Let a quantum system be described by a Hamiltonian of the form
  \begin{displaymath}
    H(t) = H+V(t) \textrm{~~{\rm on}~}\Hs
  \end{displaymath}
  where $H$ is a self-adjoint operator with a pure point spectrum and
  the spectral decomposition
  \begin{displaymath}
    H=\sum_{n\in\N} E_n P_n.
  \end{displaymath}
  Suppose that the eigen-values of $H$ are ordered increasingly and
  obey the gap condition (\ref{VGC}) with $\gamma\in\,]0,\frac12[\,$.
  Set $\alpha= 1-2\gamma$. For $p>2$ assume that
  \begin{equation}
    \label{eq:l_bound}
    \ceil{p-1} > \frac{1}{2(1-\alpha)}.
  \end{equation}
  Then there exists $\varepsilon>0$ such that if $V(t)$ is
  $T$-periodic, symmetric, continuously differentiable in the strong
  sense and obeys $\Vert{}V\Vert_{p,\gamma}\le\varepsilon$ then the
  propagator $U(t,s)$ associated to the Hamiltonian $H+V(t)$ maps
  $Q_{H}$, the form domain of $H$, onto itself and for every
  $\Psi\in{}Q_{H}$ it holds
  \begin{equation}
    \label{HPsi_t}
    \langle H\rangle_\Psi(t)
    := \langle U(t,0)\Psi,HU(t,0)\Psi\rangle = O(t^\sigma)
  \end{equation}
  where
  \begin{displaymath}
    \sigma = \frac{2\alpha}{2\ceil{p-1}(1-\alpha)-1}\,.
  \end{displaymath}
\end{theorem}

\begin{remark}
  \label{rema}
  \renewcommand{\labelenumi}{(\roman{enumi})}
  \renewcommand{\theenumi}{\roman{enumi}}
  \begin{enumerate}
  \item{There is no assumption on the dimension of $\Ran P_n$. The
      multiplicities of eigenvalues may grow arbitrarily, they can
      even be infinite.}
  \item\label{ad:VtVs}{Suppose that $V(t)\in \YY(p+1,0),$ with $p>2$,
      is $T$-periodic, symmetric, continuously differentiable in the
      strong sense and such that $[V(t),V(s)]=0$ for every $t,s$, and
      $\bar{V}:=T^{-1}\int_0^TV(t)\,\dd{}t=0$. Then one arrives at the
      same estimate (\ref{HPsi_t}). Let us outline the proof. \par
      First,
      as explained in Remark~\ref{factorised} below in which one has
      to set $r=p+1$, $Y=0$, $Z(t)=V(t)$ and $\bar{Z}=0$, one can
      transform anti-adiabatically $H+V(t)$ into $H+V_1(t)$ so that
      $V_1(t)\in\YY(p,\gamma)$ and
      \begin{displaymath}
        \|V_1\|_{p,\gamma} \leq \frac{C_H}{2C_{p+1}}
        \left(\exp(4C_{p+1}T\|V\|_{p+1,0})-1\right).
      \end{displaymath}
      Afterwards one can apply Theorem~\ref{dyn_stab} to the
      Hamiltonian $H+V_1(t)$. Arguing similarly as in the proof of
      Theorem~\ref{dyn_stab} in
      Subsection~\ref{sec:proof_thm_dyn_stab}, one finds that the
      energy diffusion for the Hamiltonian $H+V_1(t)$ is related to
      that for the Hamiltonian $H+V(t)$ by a quantity which is bounded
      in time. This is to say that estimate (\ref{HPsi_t}) holds true
      for the time evolution governed by the Hamiltonian $H+V(t)$ as
      well.}
  \item{Provided that $H(t)=H+V(t),$ with $V$ in $C^1(\R,\Bs(\Hs))$ in
      the strong sense there exists a trivial bound which does not
      depend on the spectral properties of $H$ (see \cite{N2}), namely
      \begin{equation}
        \label{nenciu}
        |\langle U(t,0)\Psi,H(t)U(t,0)\Psi \rangle|
        \le |\langle\Psi,H(0)\Psi\rangle|
        +|t|\,\sup_{s\in\R} \Vert \dot V(s)\Vert
        \Vert\Psi\Vert^2.
      \end{equation}
      For its derivation it suffices to notice that
      \begin{displaymath}
        \partial_t\langle U(t,0)\Psi,H(t)U(t,0)\Psi\rangle
        = \langle{}U(t,0)\Psi,\dot{V}(t)U(t,0)\Psi\rangle
      \end{displaymath}
      where $\dot{V}(t)$ denotes the time derivative in the strong
      sense. The estimate given by Theorem~\ref{dyn_stab} is better
      than this trivial bound if
      \begin{displaymath}
        \ceil{p-1} > p_{min}:= \frac{2\alpha+1}{2(1-\alpha)}.
      \end{displaymath}
      For example, in the case of $\alpha=2/3$ (the quantum ball) we
      get $p_{min}=7/2$. The condition $\ceil{p-1}>p_{min}$ is
      fulfilled if $p>4$ and then Theorem~\ref{dyn_stab} tells us that
      $\langle{}H\rangle_\Psi(t)=O(t^{4/5})$.}
  \item{Apart of the energy itself it is also of interest to consider
      expectation values of functions of the Hamiltonian $f(H)$ for a
      suitable choice of the function $f(\lambda)$, see for example
      \cite[Section~3]{J1}. In particular this concerns the momenta
      $H^m$, $m\in\N$. Unfortunately only several steps of our
      procedure allow for an immediate extension of this type and so
      we are not able at the moment to deal with this more general
      case. Let us discuss shortly this point. By inspection of the
      proofs of Theorem~\ref{joye} and Lemma~\ref{thm:fromYtomu} one
      finds that in both of them one can safely replace $H$ by $f(H)$
      as long as the sequence $\{f(E_n)\}$ satisfies, instead of
      $\{E_n\}$, all assumptions. The propagator $U(t,s)$ in the
      formulation of Theorem~\ref{joye} is still associated to the
      operator $H+W(t)$. The main obstacle is encountered in the proof
      of Theorem~\ref{dyn_stab} in
      Subsection~\ref{sec:proof_thm_dyn_stab}. In analogy to estimate
      (\ref{eq:tilde_average}) one can derive that
      \begin{displaymath}
        \langle\tilde U(t,0)\tilde\Psi,f(H)\tilde
        U(t,0)\tilde\Psi\rangle
        = \langle U(t,0)\Psi,J(t)f(H)J(t)^\ast U(t,0)\Psi\rangle
        = O(t^\sigma)
      \end{displaymath}
      where $\sigma$ is the same as in (\ref{eq:tilde_average}). Here
      $J(t)$ is, as detailed in the remainder of the paper, a suitable
      unitary operator constructed with the aid of the anti-adiabatic
      transform. As a next step in the proof of Theorem~\ref{dyn_stab}
      one argues that the difference $H-J(t)HJ(t)^\ast$ is bounded.
      However it does not seem to be possible to claim in general the
      same for the operator $f(H)-J(t)f(H)J(t)^\ast$. And this is
      exactly the point where the discussed extension fails.}
\end{enumerate}
\end{remark}

\subsection{Applications}

\subsubsection{The Howland's model}
\label{sec:howlands_model}

Let us apply the results of Theorem~\ref{dyn_stab} to the model
introduced by Howland in \cite{H3} and described by the Hamiltonian
$|p|^\alpha+\varepsilon{}v(\theta,t)$, with $\alpha\in\,]0,1[\,$,
which is supposed to act on $L^2(S^1,\dd\theta)$ and to be
$2\pi$-periodic in time. Set $H:=|p|^\alpha$. The spectral
decomposition of $H$ reads
\begin{displaymath}
  H = \sum_{n\ge 0} n^\alpha P_n \textrm{~~where~~}
  P_n \Psi(\theta) = \frac{1}{\pi}\int_0^{2\pi}
  \cos\left(n(\theta-s)\right)\!\Psi(s)\,\dd s.
\end{displaymath}
Except of the first one the multiplicities of the eigen-values are
equal $2$. Using integration by parts one derives that any
multiplication operator $a$ by a function $a(\theta)\in{}C^k$ obeys
the estimate
\begin{displaymath}
  \Vert P_m\,a\,P_n\Vert \leq \frac{2\sqrt{2\pi}\,
    \Vert a^{(k)}\Vert}{\langle m-n\rangle^k}\,.
\end{displaymath}
Hence $a\in\YY(k,0)$. Applying Theorem~\ref{dyn_stab} and
Remark~\ref{rema} ad~(\ref{ad:VtVs}) we get

\begin{proposition}
  Let $\alpha\in\,]0,1[$ and $v(\theta,t)$ be a real-valued function
  which is $2\pi$-periodic both in the space and in the time variable.
  Suppose that $v(\theta,t)$ is $C^k$ in $\theta $ and $C^1$ in $t$
  and such that $\int_0^{2\pi}v(\theta,t)\,\dd{}t=0$. If $k>3$ and
  $k>(5-4\alpha)/(2(1-\alpha))$ then there exists $\varepsilon_0>0$
  such that for every real $\varepsilon$,
  $|\varepsilon|<\varepsilon_0$, the propagator $U(t,s)$ associated to
  \begin{displaymath}
    H(t) := |p|^\alpha+\varepsilon v(\theta,t)\textrm{~~{\rm on}~}
    L^2(S^1,\dd\theta)
  \end{displaymath}
  preserves the domain $\Dom(|p|^{\alpha/2})$ and for every $\Psi$
  from this domain it holds true that
  \begin{displaymath}
    \langle U(t,0)\Psi,H(t)U(t,0)\Psi\rangle = O(t^\sigma)
  \end{displaymath}
  where
  \begin{equation}
    \label{eq:howland_sigma}
    \sigma = \frac{2\alpha}{2(k-2)(1-\alpha)-1}\,.
  \end{equation}
\end{proposition}

Let us summarize that the energy diffusion exponent in the Howland's
model can be made arbitrarily small provided the potential on the
circle is sufficiently smooth and the coupling constant is
sufficiently small.

\subsubsection{Discrete Hamiltonian on the half-line with a slowly
  growing potential}

Discrete models on a lattice are frequently and intensively studied.
Here we are inspired by Example~5.1 in \cite{BJ}. In contrast to it we
restrict ourselves to the usual discrete Laplacian on the half-line
rather than considering a long-range Laplacian on the line. To fit the
assumptions on which the current paper is based, in particular the gap
condition (\ref{VGC}), we further restrict ourselves to slowly growing
discrete potentials $V(n)=n^\alpha$ for some $\alpha$, $0<\alpha<1$.
Thus we do not cover the most interesting linear case $V(n)=n$.

Set $\Hs=l^2(\N)$. Let $\Delta$ be the discrete Laplacian on the
half-line,
\begin{displaymath}
  (\Delta\psi)(1)=\psi(2),\textrm{~}
  (\Delta\psi)(n)=\psi(n-1)+\psi(n+1)\textrm{~~for~}n\geq2.
\end{displaymath}
Further fix a parameter $\alpha$, $0<\alpha<1$, and define
\begin{displaymath}
  (V\psi)(n) = n^\alpha\psi(n),\textrm{~}\forall n.
\end{displaymath}
Let us consider the Hamiltonian $H(t)=-\Delta+\lambda{}a(t)V$ where
$a(t)$ is a $T$-periodic function, $a\in{}C^1$ and $a(t)\geq{}a_0>0$,
$\forall{}t\in\R$, $\lambda>0$ is a coupling constant. Set
\begin{displaymath}
  b(t) = \lambda\int_0^t a(s)\,\dd s,\textrm{~}
  \phi(t) = \frac{1}{\lambda\,a\!\left(b^{-1}(t)\right)}\,,
\end{displaymath}
and
\begin{displaymath}
  H_1(t) = V-\phi(t)\Delta.
\end{displaymath}
Hence $H(t)=\lambda{}a(t)H_1(b(t))$. Observe that $b(t)=O(t)$ and
\begin{displaymath}
  b(t+T)=b(t)+\lambda\kappa\textrm{~~where~}
  \kappa=\int_0^Ta(s)\,\dd{}s.
\end{displaymath}
Hence $b^{-1}(t+\lambda\kappa)=b^{-1}(t)+T$. The function $\phi(t)$ is
readily seen to be $C^1$ and $\lambda\kappa$-periodic. Denote by
$U_1(t,s)$ the propagator associated to $H_1(t)$. Then
$U(t,s)=U_1(b(t),b(s))$ is the propagator associated to $H(t)$.

Now one can apply Theorem~\ref{dyn_stab} to the Hamiltonian $H_1(t)$.
The unperturbed part $V$ is diagonal in the standard basis in
$l^2(\N)$, and the eigen-values obey the gap condition (\ref{VGC}). On
the other hand, the perturbation $-\phi(t)\Delta$ is strongly
differentiable and belongs to $\YY(p,0)$ for all $p\geq1$.
Theorem~\ref{dyn_stab} jointly with Remark~\ref{rema}
ad~(\ref{ad:VtVs}) implies that for any $\sigma>0$ and all initial
conditions $\Psi$,
\begin{displaymath}
  \langle U_1(t,0)\Psi,VU_1(t,0)\Psi\rangle = O(t^\sigma)
\end{displaymath}
as long as $\lambda\geq\lambda_0(\sigma)$ where $\lambda_0(\sigma)$ is
a lower bound depending on $\sigma$. Replacing $t$ by $b(t)$ one finds
that
\begin{displaymath}
  \langle U(t,0)\Psi,VU(t,0)\Psi\rangle = O(t^\sigma).
\end{displaymath}

\section{Derivation of the main result}

\subsection{Two additional theorems}

The proof of Theorem~\ref{dyn_stab} is based on the following two
theorems, Theorem~\ref{thm_aa} and Theorem~\ref{joye}. In what follows
we use the notation $D:=-\im \partial_t$ on the interval $[0,T]$ with
the periodic boundary condition.

\begin{theorem}
  \label{thm_aa}
  Let $K=D+H+V(t)$ be a Floquet Hamiltonian on $L^2([0,T],\Hs)$, with
  $H$ and $V(t)$ satisfying the assumptions of Theorem~\ref{dyn_stab}.
  Let $p>2$ and $q<p-1$ be a natural number. Then there exists
  $\varepsilon>0$ such that $\Vert{}V\Vert_{p,\gamma}\le\varepsilon$
  implies the existence of a $T$-periodic family of unitary operators
  $J(t)$ on $\Hs$ which is continuously differentiable in the strong
  sense and such that
  \begin{displaymath}
    K=J(t)(D+H+A +B(t))J(t)^*
  \end{displaymath}
  where $B(t)\in\YY(p-q,(q+1)\gamma)$ is $T$-periodic, Hermitian and
  strongly continuously differentiable, and $A$ is bounded, symmetric
  and commutes with $H$.
\end{theorem}

The remainder of the current paper is concerned with the proof of
Theorem~\ref{thm_aa}. Theorem~\ref{joye} to follow is a mere
modification of Proposition~5.1 in \cite{J2} in combination with some
ideas from \cite[Section~2]{BJ}. This is why we present its proof in a
rather sketchy form. Let us also note that the basic idea standing
behind the estimates goes back to Nenciu \cite{N2}.

\begin{theorem}
  \label{joye}
  Let $H$ be a positive operator with a pure point spectrum and the
  spectral decomposition $H=\sum_{n}E_nP_n$. Assume that the
  eigen-values $0<E_1<E_2<\dots$ satisfy $E_n=O(n^\alpha)$, with
  $\alpha>0$. Set $Q_n=1-P_n$. Let an operator-valued function
  $W(t)\in\Bs(\Hs)$ be Hermitian, $C^1$ in the strong sense and such
  that
  \begin{displaymath}
    \forall n\in\N,\textrm{~~}
    \Vert P_nW(t)Q_nH^{-1/2}\Vert
    \leq \frac\const{n^{\mu+\frac{\alpha}{2}}}
  \end{displaymath}
  uniformly in time for some $\mu>1/2$. Then the propagator $U(t,s)$
  associated with $H+W(t)$ preserves $Q_{H}$, the form domain of $H$,
  and for every $\Psi$ from $Q_{H}$,
 \begin{displaymath}
   \langle U(t,0)\Psi,HU(t,0)\Psi\rangle = O(t^{2\alpha/(2\mu-1)}).
 \end{displaymath}
\end{theorem}

\begin{rem*}
  The bound on the energy expectation value is nontrivial if
  $\mu>\frac{1}{2}+\alpha$.
\end{rem*}

\begin{proof}
  Let
  \begin{displaymath}
    W_d(t) := \sum_{n=1}^\infty P_nW(t)P_n
  \end{displaymath}
  be the diagonal part of $W(t)$. It is straightforward to see that
  $W_d(t)$ is again $C^1$ in the strong sense. Let $U_d(t,s)$ be the
  propagator associated to $H+W_d(t)$. Since $W_d(t)$ commutes with
  $H$ the same if true for $U_d(t,s)$. Equivalently this means that
  $U_d(t,s)$ commutes with all projectors $P_n$. From the Duhamel's
  formula we have
  \begin{displaymath}
    R(t) := U(t,0)-U_d(t,0) = -\im\int_0^t
    U_d(t,s)\big(W(s)-W_d(s)\big)U(s,0)\,\dd s.
  \end{displaymath}
 
  Fix $t>0$ and choose $\Psi\in\Dom(H)\subset\Dom(H^{1/2})$. Notice
  that $P_n\big(W(s)-W_d(s)\big)=P_nW(s)Q_n$. For any $t'$,
  $0\leq{}t'\leq{}t$, it holds
  \begin{displaymath}
    \|H^{1/2}U(t',0)\Psi\|^2
    = \sum_{n=1}^\infty E_n\|P_nU(t',0)\Psi\|^2
    \leq E_N\|\Psi\|^2+\sum_{n=N+1}^\infty E_n\|P_nU(t',0)\Psi\|^2.
 \end{displaymath}
  Furthermore,
  \begin{displaymath}
    \|P_nU(t',0)\Psi\|^2 \leq 2(\|P_n\Psi\|^2+\|P_nR(t')\Psi\|^2)
  \end{displaymath}
  and
  \begin{eqnarray*}
    \|P_nR(t')\Psi\| &\leq& \int_0^t\|P_nW(s)Q_nH^{-1/2}\|\,\dd s\,
    \sup_{0\leq s\leq t}\|H^{1/2}U(s,0)\Psi\| \\
    &\leq& \frac{c\,t}{n^{\mu+\frac{\alpha}{2}}}\,
    \sup_{0\leq s\leq t}\|H^{1/2}U(s,0)\Psi\|.
  \end{eqnarray*}

  From these estimates one concludes that for any $t>0$, all
  $\Psi\in\Dom(H)$, $N\in\N$ and some positive constants $c_1$, $c_2$
  independent of $t$, $\Psi$ and $N$ it holds
  \begin{displaymath}
    \left(1-\frac{c_1t^2}{N^{2\mu-1}}\right)
    \sup_{0\leq s\leq t}\|H^{1/2}U(s,0)\Psi\|^2
    \leq c_2N^\alpha\|\Psi\|^2+2\|H^{1/2}\Psi\|^2.
  \end{displaymath}
  Setting $N=[Ct^{2/(2\mu-1)}]$ where $C>0$ is a sufficiently large
  constant one deduces that there exists $c_3>0$ such that it holds
  \begin{equation}
    \label{eq:HhalfUPsi}
    \|H^{1/2}U(t,0)\Psi\|^2 \leq c_3
    \left(t^{2\alpha/(2\mu-1)}\|\Psi\|^2+\|H^{1/2}\Psi\|^2\right)
  \end{equation}
  for all $t\geq1$ and $\Psi\in\Dom(H)$.

  One can extend the validity of (\ref{eq:HhalfUPsi}) to
  $\Psi\in\Dom(H^{1/2})$. To this end it suffices to use the fact that
  $\Dom(H^{1/2})$ is a Banach space with respect to the norm
  \hspace{1em}
  $\|\Psi\|_\ast=\linebreak
  (\|\Psi\|^2+\|H^{1/2}\Psi\|^2)^{1/2}$,
  and $\Dom(H)\subset\Dom(H^{1/2})$ is a dense subspace. Choosing
  $\Psi\in\Dom(H^{1/2})$ one can find a sequence $\{\Psi_k\}$ in
  $\Dom(H)$ such that $\Psi_k\to\Psi$ in $\Dom(H^{1/2})$. Then
  (\ref{eq:HhalfUPsi}) implies that $\{U(t,0)\Psi_k\}$ is a Cauchy
  sequence in $\Dom(H^{1/2})$ whose limit necessarily equals
  $U(t,0)\Psi$. Hence $\Dom(H^{1/2})$ is $U(t,0)$--invariant and
  (\ref{eq:HhalfUPsi}) is valid also for all $\Psi\in\Dom(H^{1/2})$.
  This concludes the proof.
\end{proof}

\subsection{Proof of Theorem~\ref{dyn_stab}}
\label{sec:proof_thm_dyn_stab}

Here we show how Theorem~\ref{dyn_stab} follows from
Theorem~\ref{thm_aa} and Theorem~\ref{joye}.

\begin{lemma}
  \label{thm:fromYtomu}
  Assume that $H$ is a positive operator with a pure point spectrum
  and the spectral decomposition $H=\sum_{n=1}^\infty{}E_nP_n$, and
  such that the eigen-values satisfy $\inf{}E_nn^{-\alpha}>0$, with
  $\alpha>0$. Set $Q_n=1-P_n$.  Then for any $p\geq1$ there exist a
  constant $c(p,\alpha)>0$ such that for all $\delta>0$,
  \begin{displaymath}
    \forall B\in\YY(p,\delta),\forall n\in\N,\textrm{~~}
    \|P_nBQ_nH^{-1/2}\| \leq c(p,\alpha)\,
    \frac{\|B\|_{p,\delta}}{n^{2\delta+\frac{\alpha}{2}}}\,.
  \end{displaymath}
\end{lemma}

\begin{proof}
  Suppose that $B\in\YY(p,\delta)$. By the assumptions,
  $E_n\geq{}c\,n^\alpha$ for all $n$ and some $c>0$. We have
  \begin{displaymath}
    \|P_nBQ_nH^{-1/2}\|^2
    \leq \sum_{m,m\not=n}\frac{\|B_{n,m}\|^2}{E_m}
    \leq \frac{1}{c}\sum_{m,m\not=n}
    \frac{\|B\|_{p,\delta}^{\,2}}
    {|m-n|^{2p}\max\{m,n\}^{4\delta}m^\alpha}\,.
  \end{displaymath}
  Now one splits the range of summation in $m$ into three segments:
  $1\leq{}m<n/2$, $n/2\leq{}m<n$ and $n<m$. For each case one can
  apply elementary and rather obvious estimates to show that the
  expression decays in $n$ at least as $n^{-4\delta-\alpha}$. In the
  first case one has to use the fact that $\alpha<1$. We omit the
  details.
\end{proof}

\begin{proof}[Proof of Theorem~\ref{dyn_stab}]
  Theorem~\ref{thm_aa}, with $q:=\ceil{p-2}$, implies the existence of a
  transformation
  \begin{equation}
    \label{aa_prechod}
    K = J(t)(D+H+A+B(t))J(t)^\ast
  \end{equation}
  where $A$ is bounded and diagonal and $B(t)\in\YY(p-q,(q+1)\gamma)$.
  Since $p>2$ and $q=\ceil{p-2}$ we have $q\geq1$ and $p-q>1$. Set
  $W(t):=A+B(t)$. Then $P_nW(t)Q_n=P_nB(t)Q_n$. The gap condition
  (\ref{VGC}) guarantees that the assumptions of
  Lemma~\ref{thm:fromYtomu} are satisfied and thus one finds that
  \begin{displaymath}
    \|P_nW(t)Q_nH^{-1/2}\| = \|P_nB(t)Q_nH^{-1/2}\|
    \leq \const\cdot{}n^{-\mu-\frac{\alpha}{2}},
  \end{displaymath}
  with $\mu=2(q+1)\gamma=\ceil{p-1}(1-\alpha)$. Notice that assumption
  (\ref{eq:l_bound}) means that $\mu>1/2$. In virtue of
  Theorem~\ref{joye}, the propagator $\tilde{U}(t,s)$ associated to
  $H+W(t)$ maps the form domain $Q_{H}$ onto itself and fulfills
  \begin{equation}
    \label{eq:tilde_average}
    \langle\tilde U(t,0)\tilde\Psi, H\tilde U(t,0)\tilde\Psi\rangle
    = O(t^\sigma),\textrm{~with~}
    \sigma = \frac{2\alpha}{2\ceil{p-1}(1-\alpha)-1}\,,
  \end{equation}
  for every $\tilde\Psi\in{}Q_{H}$.

  Equality (\ref{aa_prechod}) implies that
  \begin{equation}
    \label{beda}
    H + V(t) = J(t) H J(t)^*+\im\dot{J}(t)J(t)^*+J(t)W(t)J(t)^*.
  \end{equation}
  Since the family $J(t)$ is known to be continuously differentiable
  in the strong sense it follows from the uniform boundedness
  principle that the derivative $\dot{J}(t)$ is a bounded operator.
  Moreover, using the periodicity and applying the uniform boundedness
  principle once more one finds that $\|\dot{J}(t)\|$ is bounded
  uniformly in $t$. Hence all operators occurring in equality
  (\ref{beda}), except of $H$, are bounded. One deduces from
  (\ref{beda}) that $J(t)$ maps $\Dom{}H$ onto itself for every $t$
  and that the same is also true for the form domain. Set
  $U(t,s):=J(t)\tilde U(t,s)J(s)^*$.  Then $U(t,s)$ is the propagator
  corresponding to $H+V(t)$. For any $\Psi\in{}Q_{H}$ we have
  \begin{eqnarray*}
    \langle H \rangle_\Psi(t)
    &=&\langle U(t,0)\Psi,H\,U(t,0)\Psi\rangle
    = \langle U(t,0)\Psi, J(t) H J(t)^*U(t,0)\Psi\rangle + O(1)\\
    &=& \langle \tilde U(t,0) \tilde \Psi,  H \tilde U(t,0)
    \tilde\Psi\rangle
    + O(1) = O(t^\sigma)
  \end{eqnarray*}
  where $\tilde\Psi:=J(0)^*\Psi$. This proves the theorem.
\end{proof}

\subsection{The idea of the proof of Theorem~\ref{thm_aa}}
\label{sec:idea}

It remains to prove Theorem~\ref{thm_aa}. The proof is somewhat
lengthy and the remainder of the paper is devoted to it. Let us
explain the main idea. The proof combines the anti-adiabatic
transformation due to Howland (see Section~\ref{section_aa}) with a
(properly modified) diagonalization method, as presented in
\cite{DLSV} (see Section~\ref{sec:diagonalization}). This procedure is
applied repeatedly until achieving the required properties of the
perturbation. Let us describe one step in this approach when starting
from the Floquet Hamiltonian
\begin{displaymath}
  K_\triangle := D+H+Y+Z(t)
\end{displaymath}
where $Y\in\YY(\infty,\gamma)$ is Hermitian and diagonal (i.e.,
commuting with $H$) and $Z(t)\in\YY(r,i\gamma)$ is symmetric,
$T$-periodic and strongly $C^1$. The parameters are supposed to
satisfy $i\ge1$, $r>2$.

Firstly, using the anti-adiabatic transform we try to improve the
decay of entries of $Z(t)$ along the main diagonal when paying for it
by a worse decay of elements in the direction perpendicular to the
diagonal. In more detail, we would like to transform
$Z(t)\in\YY(r,i\gamma)$ into
$Z_\diamondsuit(t)\in\YY(r-1,(i+1)\gamma)$. Unfortunately, we are not
able to get rid of the extra term $\bar{Z}\in\YY(r,i\gamma)$, the time
average of $Z(t)$. The anti-adiabatic transform can be schematically
described as
\begin{displaymath}
  K_\triangle = D+H+Y+Z(t)\rightarrow K_\diamondsuit
  = D+H+Y+\bar Z+Z_\diamondsuit(t).
\end{displaymath}
To cope with the unwanted extra term we apply afterwards a
diagonalization procedure which in fact means the transform
\begin{displaymath}
  K_\diamondsuit = D+H+Y+\bar Z+Z_\diamondsuit(t)
  \rightarrow K_\heartsuit := D+H+A+B(t)
\end{displaymath}
where $A$ and $B(t)$ already have the desired properties, i.e.,
$B(t)\in\YY(r-1,(i+1)\gamma)$ is symmetric, $T$-periodic and strongly
$C^1$, and $A\in\YY(\infty,\gamma)$ is Hermitian and commuting with
$H$.

\section{The anti-adiabatic transform}
\label{section_aa}
 
In this section we adapt the strategy of Howland \cite{H3} and make
precise the mapping $K_\triangle\rightarrow{}K_\diamondsuit$, as
announced in Subsection~\ref{sec:idea}. Using the anti-adiabatic
transform, i.e., roughly speaking, by applying the commutator with $H$
one can improve the decay of matrix entries of the perturbation along
the main diagonal at the expense of a slower decay in the direction
perpendicular to the diagonal. Using the language of classes
$\YY(p,\delta)$, the anti-adiabatic transform may be viewed as passing
from a perturbation $Z(t)\in\YY(p,\delta)$ to a new perturbation
$Z_1(t)\in\YY(p-1,\delta+\gamma)$ where $\gamma$ comes from the gap
condition (\ref{VGC}) (see Lemma~\ref{how}).

Let us introduce the transform in detail. Let $K_\triangle$ be a
Floquet Hamiltonian of the form
\begin{displaymath}
  K_\triangle = D+H+Y+Z(t),
\end{displaymath}
with $H$ satisfying the assumptions of Theorem~\ref{dyn_stab},
$Y\in\YY(\infty,\gamma)$ being Hermitian and commuting with $H$, and
$Z(t)\in\YY(r,i\gamma)$ being Hermitian, $T$-periodic and continuous
in the strong sense. By the uniform boundedness principle, $\|Z(t)\|$
is bounded uniformly in $t$. The parameters are supposed to satisfy
$r>2$, $i\ge1$. Set
\begin{displaymath}
  \bar{Z} := \frac{1}{T}\int_0^TZ(t)\,\dd{}t,\textrm{~~}
  \tilde{Z}(t) = Z(t)- \bar{Z}.
\end{displaymath}
Define
\begin{displaymath}
  F(t):=  \int_0^t \tilde Z(s)\,\dd s,
\end{displaymath}
so that $F(t)$ is Hermitian, $T$-periodic, strongly $C^1$ and lying in
$\YY\left(r,i\gamma\right)$. Let us define $K_\diamondsuit$ by the
gauge-type transformation of $K_\triangle$,
\begin{displaymath}
  K_\diamondsuit:=e^{\im F(t)} K_\triangle e^{-\im F(t)}
  = D+H+Y+\bar Z +Z_\diamondsuit(t) ,
\end{displaymath}
with
\begin{equation}
 \label{Z_kara}
 Z_\diamondsuit(t)=e^{\im F(t)}\left(D+H+Y+Z(t)\right)e^{-\im F(t)}
 -\left(D+H+Y+\bar Z \right).
\end{equation}

The main result related to the anti-adiabatic transform is as follows.

\begin{proposition}
  \label{a-a}
  Let $r>2$, $i\ge1$, $\gamma\in\,]0,\frac12[$, and $H$ be a
  self-adjoint operator with a pure point spectrum and the spectral
  decomposition $H=\sum_{n}E_nP_n$. Assume that the eigen-values
  $\{E_n\}_{n=1}^\infty$ are ordered increasingly and satisfy the
  inequality
  \begin{displaymath}
    |E_m-E_n| \leq C_H\frac{|m-n|}{\max\{m,n\}^{2\gamma}}.
  \end{displaymath}
  Furthermore, $Y$ and $Z(t)$ obey the assumptions formulated above.
  
  Then $Z_\diamondsuit(t)$ defined in (\ref{Z_kara}) is $T$-periodic,
  continuous in the strong sense, Hermitian, and lies in
  $\YY\left(r-1,(i+1)\gamma\right)$. The norm of $Z_\diamondsuit$
  obeys the bound
  \begin{equation}
    \label{Z_diamond_estim}
    \Vert Z_\diamondsuit\Vert_{r-1,(i+1)\gamma}
    \leq \frac{\exp\!\left(4C_rT\,\Vert Z\Vert_{r,i\gamma}\right)-1}
    {2C_r}
    \left(C_H+4\Vert Y\Vert_{\infty,\gamma}
      +2C_r\Vert Z\Vert_{r,i\gamma}\right),
  \end{equation}
  with the constant $C_r$ defined in (\ref{norm_const}). The
  operator-valued function $e^{\im F(t)}$ is $C^1$ in the strong
  sense. Moreover, if $Z(t)$ is $C^1$ in the strong sense then the
  same is true for $Z_\diamondsuit(t)$.
\end{proposition}

\begin{proof}
  The periodicity and the differentiability are clear from the above
  discussion. The RHS of (\ref{Z_kara}) can be expanded according to
  the formula
  \begin{displaymath}
    e^A B e^{-A} = B +\sum_{j=1}^\infty \frac1{j!}\ad_A^j(B).
  \end{displaymath}
  Here we use the notation $\ad_A(B):=[A,B]=AB-BA$. Since
  $\ad_{F(t)}D=\im\dot F(t)=\im\tilde{Z}(t)$ we get
  \begin{eqnarray}
    \nonumber
    Z_\diamondsuit(t) &=& \sum_{j=1}^\infty \frac{\im^j}{j!}\,
    \ad^{\,j-1}_{F(t)}\left(\im\tilde Z(t)+\left[F(t),H+Y+Z(t)\right]
    \right) +\tilde Z(t) \\
    \label{V_nplus1}
    &=& \sum_{j=1}^\infty \frac{\im^j}{j!}\, \ad^{\,j-1}_{F(t)}X(t)
  \end{eqnarray}
  where
  \begin{displaymath}
    X(t) := \ad_{F(t)}\!\left(H+Y+Z(t)-\frac1{j+1}\tilde Z(t)\right) 
    = \ad_{F(t)}\!\left(H+Y+\frac j{j+1}Z(t)+\frac1{j+1}\bar Z\right).
  \end{displaymath}
  By Lemma~\ref{how}, $\ad_{F(t)}H\in\YY(r-1,(i+1)\gamma)$, and
  according to Corollary~\ref{lemma_XY}, the same holds true for
  $\ad_{F(t)}Z(t)$ and $\ad_{F(t)}\bar{Z}$. Notice also that
  $\|\bar{Z}\|_{p,\delta}\leq\|Z\|_{p,\delta}$. Furthermore, since
  $Y\in\YY(\infty,\gamma)$ is diagonal we have
  \begin{eqnarray*}
    && \langle m-n\rangle^{r-1}\max\{m,n\}^{2(i+1)\gamma}
    \|(F(t)Y)_{m,n}\|\\
    && \leq \frac{1}{\langle m-n\rangle}
    \left(\frac{\max\{m,n\}}{n}\right)^{\!2\gamma}
    n^{2\gamma}\|F\|_{r,i\gamma}\|Y_{n,n}\|
    \leq 2^{2\gamma}\|F\|_{r,i\gamma}\|Y\|_{\infty,\gamma}.
  \end{eqnarray*}
  Hence
  $\|F(t)Y\|_{r-1,(i+1)\gamma}\leq2\|F\|_{r,i\gamma}
  \|Y\|_{\infty,\gamma}$. The same estimate is true for
  $\|YF(t)\|_{r-1,(i+1)\gamma}$ and therefore
  $\|\ad_{F}Y\|_{r-1,(i+1)\gamma}\leq4\|F\|_{r,i\gamma}
  \|Y\|_{\infty,\gamma}$. We conclude that $X(t)$ belongs to
  $\YY(r-1,(i+1)\gamma)$ and
  \begin{equation}
    \label{eq:X_estim}
    \Vert X\Vert_{r-1, (i+1)\gamma} \leq \Vert F\Vert_{r,i\gamma}
    \left(C_H+4\Vert Y\Vert_{\infty,\gamma}
      +2C_r\Vert Z\Vert_{r,i\gamma}\right).
  \end{equation}
  Recalling Corollary~\ref{lemma_XY} once more we have
  \begin{displaymath}
    \YY(r-1,(i+1)\gamma)\YY(r,i\gamma),\textrm{~}
    \YY(r,i\gamma)\YY(r-1,(i+1)\gamma) \subset \YY(r-1,(i+1)\gamma)
  \end{displaymath}
  and so $\ad^{\,j-1}_{F(t)}X(t)$ lies in $\YY(r-1,(i+1)\gamma)$ as
  well and
  \begin{equation}
    \label{eq:adFX_estim}
    \Vert\ad^{\,j-1}_{F}X\Vert_{r-1,(i+1)\gamma}
    \leq \left(2C_r \Vert F\Vert_{r,i\gamma}\right)^{j-1}
    \Vert X\Vert_{r-1,(i+1)\gamma}\,.
  \end{equation}
  Consequently, the series (\ref{V_nplus1}) converges in the Banach
  space $\YY\left(r-1,(i+1)\gamma\right)$. To derive inequality
  (\ref{Z_diamond_estim}) from (\ref{eq:X_estim}) and
  (\ref{eq:adFX_estim}) one applies the estimate
  $\Vert{}F\Vert_{r,i\gamma}\le2T\Vert Z\Vert_{r,i\gamma}$ which
  immediately follows from the definition of $F(t)$ and
  $\tilde{Z}(t)$. This completes the proof.
\end{proof}

\begin{remark}
  \label{factorised}
  The proposition holds also true for $i=0$ provided $[Z(t),Z(s)]=0$
  for every $t,s$. In this case $F(t)$ commutes with $Z(t)$ and
  $\bar{Z}$, and the formula (\ref{V_nplus1}) holds true with
  $X(t)=\ad_{F(t)}(H +Y)$. Repeating the steps from the proof of the
  proposition one arrives at the inequality
  \begin{displaymath}
    \Vert Z_\diamondsuit\Vert_{r-1,(i+1)\gamma}
    \le \frac{\exp\!\left(4C_rT\,\Vert Z\Vert_{r,i\gamma}\right)-1}
    {2C_r}\,(C_H+2\Vert Y\Vert_{\infty,\gamma}).
  \end{displaymath}
\end{remark}

\section{The diagonalization procedure}
\label{sec:diagonalization}

\subsection{Formulation of the result}
\label{sec:formulation}

The main result of this section is formulated in the following
proposition.

\begin{proposition}
  \label{p-d}
  Let $i\ge1$ be a natural number, $\gamma\in\,]0,\frac12[\,$, and $H$
  be a self-adjoint operator with a pure point spectrum and the
  spectral decomposition $H=\sum_{n}E_nP_n$. Assume that the
  eigen-values $\{E_n\}_{n=1}^\infty$ are ordered increasingly and
  satisfy the inequality
  \begin{equation}
    \label{c_H}
    |E_m-E_n| \ge c_H\frac{|m-n|}{\max\{m,n\}^{2\gamma}}.
  \end{equation}
  Let $Y\in\YY(\infty, \gamma)$ be Hermitian and commuting with $H$.
  Suppose that $\bar Z$ is Hermitian and belongs to the class
  $\YY(r,i\gamma)$ for some $r>2$. Finally, assume that
  \begin{equation}
    \label{normbarZ}
    \|Y\|_{\infty,\gamma}
    +\Vert\bar Z\Vert_{r,i\gamma}
    \le \frac{c_H}{4\pi\,C_{r+1}},
  \end{equation}
  with the constant $C_{r+1}$ given by (\ref{norm_const}).
  
  Then there exists $U$, a unitary operator on $\Hs$, such that
  \begin{equation}
    \label{eq:UinftyH}
    U(H+Y+\bar Z)U^* = H+A
  \end{equation}
  where $A\in\YY(\infty,\gamma)$ commutes with $H$ and obeys
  \begin{equation}
    \label{eq:normA_estim}
    \Vert{}A\Vert_{\infty,\gamma}
    \le 2\left(\Vert{}Y\Vert_{\infty,\gamma}
      +\Vert\bar{Z}\Vert_{r,i\gamma}\right).
  \end{equation}
  Moreover, for every operator $X\in\YY(r-1,(i+1)\gamma)$ it holds
  \begin{equation}
    \label{eq:UXU_estim}
    \Vert UXU^*\Vert_{r-1,(i+1)\gamma}
    \leq \exp\!\left(2\,\frac{C_r}{C_{r+1}}\right)
    \Vert X\Vert_{r-1,(i+1)\gamma}.
  \end{equation}
\end{proposition}

Since $U$ does not depend on time this result can be interpreted in
the following way.

\begin{corollary}
  \label{thm:Kheart}
  Let us consider a Floquet Hamiltonian of the form
  \begin{displaymath}
    K_\diamondsuit = D+H+Y+\bar{Z}+Z_\diamondsuit(t)
  \end{displaymath}
  where $H$, $Y$ and $\bar{Z}$ obey the same assumptions as in
  Proposition~\ref{p-d}, with $r>2$ and $i\geq1$, and
  $Z_\diamondsuit(t)\in\YY\left(r-1,(i+1)\gamma\right)$ is
  $T$-periodic, continuously differentiable in the strong sense and
  Hermitian.
  
  Then there exists a unitary operator $U$ on $\Hs$ such that for the
  transformed Floquet Hamiltonian
  \begin{displaymath}
    K_\heartsuit := UK_\diamondsuit U^* = D+H+A+B(t)
  \end{displaymath}
  it holds: $A\in\YY(\infty,\gamma)$ commutes with $H$ and fulfills
  (\ref{eq:normA_estim}),
  \begin{displaymath}
    B(t) := UZ_\diamondsuit(t)U^*
    \in\YY(r-1,(i+1)\gamma)
  \end{displaymath}
  is $T$-periodic, continuously differentiable in the strong sense,
  Hermitian and satisfies
  \begin{displaymath}
    \Vert B\Vert_{r-1,(i+1)\gamma}
    \leq \exp\!\left(2\,\frac{C_r}{C_{r+1}}\right)
    \Vert Z_\diamondsuit\Vert_{r-1,(i+1)\gamma}.
  \end{displaymath}
\end{corollary}

The proof of Proposition~\ref{p-d} is a modification (to the case of
shrinking gaps) of a diagonalization procedure introduced in
\cite{DLSV} and conventionally called the progressive diagonalization
method.

\subsection{The algorithm}

The diagonalization procedure is constructed iteratively, let us first
describe the algorithm. Starting from $H+Y+\bar Z$ we construct the
first 4-tuple of operators
\begin{displaymath}
  U_0 := 1,\textrm{~} G_1 := Y+\diag\bar Z,
  \textrm{~} V_1 := \offdiag \bar Z,\textrm{~} H_1 := H+G_1+V_1,
\end{displaymath}
where 
\begin{displaymath}
  \diag X := \sum_{n\in\N} P_n X P_n,\textrm{~}
  \offdiag X := \sum_{m\neq n} P_m X P_n
\end{displaymath}
denote the diagonal and the off diagonal part of the matrix of an
operator $X$ with respect to the eigen-basis of $H$. We define
recursively a sequence of operators $H_s$, $G_s$, $V_s$, $W_s$ and
$U_s$ by the following rules: provided $G_s$ and $V_s$ have been
already defined let $W_s$ be the solution of
\begin{equation}
  \label{anticom}
  [H+G_s,W_s] = V_s \textrm{~~and~~} \diag W_s = 0.
\end{equation}
We define
\begin{equation}
  \label{Ha_s0}
  H_{s+1} := e^{W_s}H_s e^{-W_s}.
\end{equation}
Finally, we set 
\begin{equation}
  \label{eq:UsGsVs_def}
  U_s := e^{W_s} U_{s-1},\textrm{~}
  G_{s+1} := \diag H_{s+1}-H,\textrm{~}
  V_{s+1} := \offdiag H_{s+1}.
\end{equation}

Since $H_s=H+G_s+V_s$ for all $s$ and with the aid of (\ref{anticom})
one derives from (\ref{Ha_s0}) that
\begin{eqnarray}
  \nonumber
  H_{s+1} &=& H_s+\sum_{k=1}^\infty\frac{1}{k!}\,
  \ad^{\,k-1}_{W_s}[W_s,H_s]
  \,=\, H+G_s+V_s+\sum_{k=1}^\infty\frac{1}{k!}\,
  \ad^{\,k-1}_{W_s}(-V_s+[W_s,V_s])\\
  \label{Ha_s}
  &=& H+G_s+\Phi(\ad_{W_s})V_s
\end{eqnarray}
where
\begin{equation}
  \label{eq:Phi_def}
  \Phi(x) := \sum_{k=1}^\infty\frac k{(k+1)!}\,x^k
  = e^x-\frac1x(e^x-1)
\end{equation}

Observe also that in the course of the algorithm, $G_s$ is always
diagonal (commuting with $H$) and symmetric, $V_s$ is symmetric and
off diagonal, $W_s$ is antisymmetric and off diagonal. Therefore
$e^{W_s}$ and $U_s$ are unitary. It is straightforward to prove by
induction that for every $s=1,2,\ldots$,
\begin{equation}
  \label{Hs_plus_1}
  H+G_{s+1}+V_{s+1}=U_s(H+Y+\bar Z)U_s^*.
\end{equation}

\subsection{Auxiliary facts}
\label{sec:auxiliary}

To solve the commutator equation (\ref{anticom}) we need the following
result taken from a paper by Bhatia and Rosenthal.

\begin{lemma}[\cite{BR}]
  \label{Bhatia}
  Let $E$ and $F$ be two Hilbert spaces. Let $A$ and $B$ be Hermitian
  operators (i.e., bounded and self-adjoint) on $E$ and $F$,
  respectively, such that $\dist\!\left(\sigma(A),\sigma(B)\right)>0$.
  Then for every bounded operator $Y:F\rightarrow{}E$ there exists a
  unique bounded operator $X:F\rightarrow{}E$ such that
  \begin{displaymath}
    A X- X B=Y. 
  \end{displaymath}
  Moreover, the inequality
  \begin{displaymath}
    \Vert X\Vert
    \le \frac\pi{2\dist\!\left(\sigma(A),\sigma(B)\right)}\,
    \Vert Y\Vert,
  \end{displaymath}
  holds true.
\end{lemma}

\begin{rem*}
  The solution $X$ is given by
  \begin{displaymath}
    X=\int_{\R} e^{-\im t A} Y e^{\im t B}f(t)\,\dd t
  \end{displaymath}
  for any $f\in L^1(\R)$ such that its Fourier image obeys
  $\hat{f}(s)=1/\sqrt{2\pi}s$ on the set
  \linebreak$\sigma(A)-\sigma(B)$.
  This implies $\Vert{}X\Vert\le\Vert{}f\Vert_1\Vert{}Y\Vert$, and
  optimizing over such $f$ one gets the constant $\pi/2$.
\end{rem*}

In the algorithm plays a certain role the function $\Phi(x)$
introduced in (\ref{eq:Phi_def}). It is supposed to be defined on the
interval $[0,\infty[$. Let us point out here some of its elementary
properties. This is a strictly increasing function mapping the
interval $[0,\infty[$ onto itself. It holds $\Phi(0)=0$, $\Phi(1)=1$,
and so the function maps also the interval $]0,1[$ onto itself.
Moreover, $\Phi(x)$ is a convex function and so
\begin{equation}
  \label{eq:Phi_convex}
  \forall x\in\,]0,1[,\textrm{~}
  \Phi(x) < x.
\end{equation}

Further, let us consider a sequence $\{x_s\}_{s=1}^\infty$ formed by
nonnegative numbers obeying the inequalities
\begin{equation}
  \label{eq:x_leq_Phix}
  \forall s\in\N,\textrm{~}x_{s+1} \leq \Phi(x_s)x_s.
\end{equation}
If $x_1<1$ then the sequence is non-increasing and
(\ref{eq:Phi_convex}), (\ref{eq:x_leq_Phix}) imply that
$x_{s+1}\leq{}x_s^{\,2}$. It follows that
\begin{displaymath}
  \forall s\in\N,\textrm{~}x_s \leq x_1^{\,2^{s-1}},
\end{displaymath}
and
\begin{equation}
  \label{eq:xs_sum}
  \sum_{s=1}^\infty x_s \leq \frac{x_1}{1-x_1} < \infty.
\end{equation}

\subsection{Convergence of the algorithm}

\begin{proof}[Proof of Proposition~\ref{p-d}]
  We have to prove that $V_s\to0$, $G_s\to{}A$ and $U_s\to{}U$. The
  key ingredient of the algorithm is the control of the size of $W_s$
  given as the off diagonal solution to the commutator equation
  (\ref{anticom}). For every $m\neq{}n$ we seek $W_s(m,n)$ such that
  \begin{displaymath}
    \left(E_m +(G_s)_{m,m}\right) (W_s)_{m,n}-(W_s)_{m,n}
    \left(E_n+(G_s)_{n,n}\right) = (V_s)_{m,n}.
  \end{displaymath}
  Suppose for the moment that $G_s$ lies in $\YY(\infty,\gamma)$ for
  every $s\in\N$ with
  \begin{equation}
    \label{G}
    \Vert G_s \Vert_{\infty,\gamma} \le \frac{c_H}{6}.
  \end{equation}
  The norm $\|\cdot\|_{\infty,\gamma}$ makes sense in this case since
  $G_s$ is diagonal for every $s\in\N$. The spectrum of
  $E_n+(G_s)_{n,n}$ is a subset of the interval
  \begin{displaymath}
    \Big[\,E_n-\frac{\Vert G_s\Vert_{\infty,\gamma}}{n^{2\gamma}},
    E_n+\frac{\Vert G_s \Vert_{\infty,\gamma}}{n^{2\gamma}}\,\Big]. 
  \end{displaymath}
  Owing to (\ref{c_H}) the distance between the spectrum of
  $E_m+(G_s)_{m,m}$ and $E_n+(G_s)_{n,n}$ can be estimated from below
  by
  \begin{eqnarray}
    \nonumber
    |E_m-E_n|-\Vert G_s \Vert_{\infty,\gamma}
    \left(m^{-2\gamma}+n^{-2\gamma} \right)
    &\geq& c_H\,\frac{|m-n|}{\max\{m,n\}^{2\gamma}}
    -\frac{c_H}{6}(m^{-2\gamma}+n^{-2\gamma}) \\
    \label{small_divisors}
    &\geq& \frac{c_H |m-n|}{2\max\{m,n\}^{2\gamma}}\,.
  \end{eqnarray}
  The last inequality in (\ref{small_divisors}) is a consequence of
  the following estimate where we assume for definiteness that $m>n$
  (recall that $2\gamma<1$):
  \begin{displaymath}
    \frac{3(m-n)}{m^{2\gamma}}
    \geq m^{-2\gamma}+\frac{m}{n}\,m^{-2\gamma}
    \geq m^{-2\gamma}+n^{-2\gamma}.
  \end{displaymath}
  Applying Lemma~\ref{Bhatia} we conclude that
  \begin{equation}
    \label{eq:Wmn_leq_Vmn}
    \Vert(W_s)_{m,n}\Vert
    \le \frac{\pi\max\{m,n\}^{2\gamma}}{c_H |m-n|}
    \Vert (V_s)_{m,n}\Vert.
  \end{equation}

  Set
  \begin{equation}
    \label{eq:M_xs_def}
    M:=\frac{c_H}{2\pi C_{r+1}},\textrm{~~}
    x_s := \frac{\Vert V_s\Vert_{r,i\gamma}}{M},
  \end{equation}
  If $V_s$ lies in the class $\YY(r,i\gamma)$ then one derives from
  (\ref{eq:Wmn_leq_Vmn}) that $W_s\in\YY(r+1,(i-1)\gamma)$ and
  \begin{equation}
    \label{eq:Ws_leq_xs}
    \Vert W_s \Vert_{r+1,(i-1)\gamma}\le 
    \frac \pi{c_H}\Vert V_s\Vert_{r,i\gamma} = \frac{x_s}{2C_{r+1}}.
  \end{equation}
  From Corollary~\ref{lemma_XY} it follows that
  $\ad^{\,k}_{W_s}V_s\in\YY(r,i\gamma)$ and
  \begin{equation}
    \label{kolotoc}
    \Vert \ad^{\,k}_{W_s}V_s\Vert_{r,i\gamma}
    \le \left( 2 C_{r+1}
      \Vert W_s\Vert _{r+1, (i-1)\gamma} \right)^k 
    \Vert V_s\Vert_{r,i\gamma}
    \leq x_s^{\,\,k}\|V_s\|_{r,i\gamma},
  \end{equation}
  Since $V_{s+1}$ is defined as the off diagonal part of $H_{s+1}$ we
  get from (\ref{Ha_s}) and (\ref{kolotoc}) that
  \begin{displaymath}
    V_{s+1} = \offdiag\!\left(\Phi(\ad_{W_s})V_s\right).
  \end{displaymath}
  and so
  \begin{displaymath}
    \Vert V_{s+1}\Vert_{r, i\gamma}
    \le \Phi(x_s)\Vert V_s\Vert_{r,i\gamma}.
  \end{displaymath}
  Hence the sequence $\{x_s\}$ defined in (\ref{eq:M_xs_def}) fulfills
  inequalities (\ref{eq:x_leq_Phix}).
  
  Since $\|V_1\|_{r,i\gamma}\leq\|\bar{Z}\|_{r,i\gamma}$ assumption
  (\ref{normbarZ}) implies $x_1\leq1/2$. We know from the discussion
  at the end of Subsection~\ref{sec:auxiliary} that in that case the
  series $\sum{}x_s$ is convergent. It follows that
  $\Vert{}V_s\Vert_{r,i\gamma}\to0$ and, using the estimate
  \begin{displaymath}
    \Vert W_s \Vert\le \Vert W_s \Vert_{SH}
    \le \left(1+2\zeta(r+1)\right)\Vert W_s\Vert_{r+1,(i-1)\gamma}
  \end{displaymath}
  and (\ref{eq:Ws_leq_xs}), also that $U_s$ converges to a unitary
  operator $U$ in $\Bs(\Hs)$. Furthermore, from (\ref{Ha_s}) and
  (\ref{eq:UsGsVs_def}) one deduces that
  \begin{displaymath}
    G_{s+1}-G_s = \diag\!\left(\Phi(\ad_{W_s})V_s\right).
  \end{displaymath}
  Since $G_s$ is diagonal and $i\geq1$ we have
  \begin{displaymath}
    \|G_{s+1}-G_s\|_{\infty,\gamma} = \|G_{s+1}-G_s\|_{r,\gamma}
    \leq \|G_{s+1}-G_s\|_{r,i\gamma}
    \leq \|\Phi(\ad_{W_s})V_s\|_{r,i\gamma}.
  \end{displaymath}
  Using once more (\ref{eq:Ws_leq_xs}) and (\ref{kolotoc}) one finds
  that
  \begin{eqnarray}
    \label{eq:Gsdiff_estim}
    \|G_{s+1}-G_s\| = \|G_{s+1}-G_s\|_{\infty,0}
    \leq \|G_{s+1}-G_s\|_{\infty,\gamma} \leq M \Phi(x_s)x_s.
  \end{eqnarray}
  From here one concludes that $\{G_s\}$ is a Cauchy sequence both in
  $\YY(\infty,\gamma)$ and $\Bs(\Hs)$. Hence $G_s$ converges to a
  diagonal operator $A$ which lies in $\YY(\infty,\gamma)$.
  
  We must verify that condition (\ref{G}) is actually fulfilled.
  Observe from (\ref{norm_const}) that $C_p>2^3\cdot3$ if $p>2$. By
  the assumptions,
  \begin{displaymath}
    \Vert G_1\Vert_{\infty,\gamma} \le \Vert Y\Vert_{\infty,\gamma}
    +\Vert\bar{Z}\Vert_{r, i\gamma} < \frac{c_H}{12}.
  \end{displaymath}
  Furthermore, from (\ref{eq:Gsdiff_estim}) it follows that
  \begin{equation}
    \label{eq:Gs_estim}
    \|G_{s+1}\|_{\infty,\gamma}
    \leq \|G_1\|_{\infty,\gamma}
    +\sum_{j=1}^s\|G_{s+1}-G_s\|_{\infty,\gamma}
    \leq \frac{c_H}{12}+M\sum_{j=1}^\infty x_j\Phi(x_j).
  \end{equation}
  Recalling that $x_1\leq1/2$ one gets
  \begin{equation}
    \label{eq:sum_xPhi_estim}
    M\sum_{j=1}^\infty x_j\Phi(x_j) \leq \frac{Mx_1^{\,2}}{1-x_1}
    \leq Mx_1 \leq \|\bar Z\|_{r,i\gamma} < \frac{c_H}{12}.
  \end{equation}
  The last inequality is again a consequence of assumption
  (\ref{normbarZ}). One concludes that condition (\ref{G}) is
  fulfilled for all $s$.
  
  Since all operators occurring in (\ref{Hs_plus_1}) except of $H$ are
  bounded one deduces from this equality that $U_s$ preserves the
  domain of $H$ for all $s$. Since $H$ is a closed operator the limit
  in equality (\ref{Hs_plus_1}), as $s\to\infty$, can be carried out
  and results in equality (\ref{eq:UinftyH}).

  From the computations in (\ref{eq:Gs_estim}),
  (\ref{eq:sum_xPhi_estim}) it also follows that
  \begin{displaymath}
    \|G_{s+1}\|_{\infty,\gamma} \leq \|G_1\|_{\infty,\gamma}+Mx_1
    = \|G_1\|_{\infty,\gamma}+\|V_1\|_{r,i\gamma}
    \leq \|Y\|_{\infty,\gamma}+2\|\bar Z\|_{r,i\gamma}.
  \end{displaymath}
  Sending $s$ to infinity one verifies the estimate
  (\ref{eq:normA_estim}). Furthermore, estimate (\ref{eq:Ws_leq_xs})
  implies
  \begin{displaymath}
    \sum_{s=1}^\infty\Vert W_s\Vert_{r+1,(i-1)\gamma}
    \leq \frac{1}{2C_{r+1}}\sum_{s=1}^\infty x_s
    \leq \frac{x_1}{2C_{r+1}(1-x_1)} \leq \frac{1}{2C_{r+1}}.
  \end{displaymath}
  From Corollary~\ref{lemma_XY} we deduce that the operator
  $\ad_{W_s}$ is well defined on the Banach space
  $\YY(r-1,(i+1)\gamma)$, with a norm bounded from above by
  $4C_r\|W_s\|_{r+1,(i-1)\gamma}$. Thus for $X\in\YY(r-1,(i+1)\gamma)$
  one can estimate
  \begin{eqnarray*}
    \Vert UXU^*\Vert_{r-1,(i+1)\gamma} &=& \lim_{s\to\infty}
    \|e^{W_s}e^{W_{s-1}}\cdots
    e^{W_1}Xe^{-W_1}\cdots e^{-W_{s-1}}e^{-W_s}\|_{r-1,(i+1)\gamma} \\
    &\leq& \exp\!\left(4C_r\sum_{s=1}^\infty
      \Vert W_s\Vert_{r+1,(i-1)\gamma}
    \right)\Vert X\Vert_{r-1,(i+1)\gamma} \\
    &\leq& \exp\!\left(2\,\frac{C_r}{C_{r+1}}\right)
    \Vert X\Vert_{r-1,(i+1)\gamma}.
  \end{eqnarray*}
  This shows (\ref{eq:UXU_estim}). The proof is complete.
\end{proof}

\section{Proof of Theorem~\ref{thm_aa}}

As already announced, the proof of Theorem~\ref{thm_aa} is based on a
combination of the anti-adiabatic transform (Proposition~\ref{a-a})
and the progressive diagonalization method
(Corollary~\ref{thm:Kheart}). Let us formulate it as a corollary.

\begin{corollary}
  \label{thm:combination}
  Let $r>2$, $i\ge1$, $\gamma\in\,]0,\frac12[\,$, and $H$ be a
  self-adjoint operator with a pure point spectrum and the spectral
  decomposition $H=\sum_{n}E_nP_n$. Assume that the eigen-values
  $\{E_n\}_{n=1}^\infty$ are ordered increasingly and satisfy
  (\ref{VGC}). Further assume that $Y\in\YY(\infty,\gamma)$ is
  Hermitian and commutes with $H$, and $Z(t)\in\YY(r,i\gamma)$ is
  Hermitian, $T$-periodic and $C^1$ in the strong sense. If
  \begin{displaymath}
    \|Y\|_{\infty,\gamma}+\Vert Z\Vert_{r,i\gamma}
    \le \frac{c_H}{4\pi\,C_{r+1}}
  \end{displaymath}
  then there exists a family $\UU(t)$ of unitary operators on $\Hs$
  which is $T$-periodic and $C^1$ in the strong sense and such that
  \begin{displaymath}
    \UU(t)\left(D+H+Y+Z(t)\right)\UU(t)^* = D+H+A+B(t)
  \end{displaymath}
  where $A\in\YY(\infty,\gamma)$ is Hermitian, commutes with $H$ and
  fulfills
  \begin{displaymath}
    \Vert{}A\Vert_{\infty,\gamma}
    \leq 2\left(\Vert Y\Vert_{\infty,\gamma}
      +\Vert Z\Vert_{r,i\gamma}\right),
  \end{displaymath}
  and $B(t)\in\YY(r-1,(i+1)\gamma)$ is $T$-periodic, Hermitian,
  continuously differentiable in the strong sense and satisfies
  \begin{eqnarray*}
    \Vert B\Vert_{r-1,(i+1)\gamma}
    &\leq& \frac{1}{2C_r}\exp\!\left(2\,\frac{C_r}{C_{r+1}}\right) \\
    && \textrm{~}\times\,
    \big(\exp\!\left(4C_rT\,\Vert Z\Vert_{r,i\gamma}\right)-1\big)
    \left(C_H+4\Vert Y\Vert_{\infty,\gamma}
      +2C_r\Vert Z\Vert_{r,i\gamma}\right).
  \end{eqnarray*}
\end{corollary}

To prove Corollary~\ref{thm:combination} it suffices to set
$\UU(t)=U\exp(\im{}F(t))$ where $F(t)$ comes from
Proposition~\ref{a-a} and $U$ comes from Corollary~\ref{thm:Kheart}.
Apart of this one applies the following elementary estimate: if the
norm $\|X\|_{p,\delta}$ of a $T$-periodic family $X(t)$ formed by
bounded operators is finite for some $p>1$ and $\delta\geq0$ then the
time average $\bar{X}$ of $X(t)$ over the period $T$ fulfills
$\|\bar{X}\|_{p,\delta}\leq\|X\|_{p,\delta}$.

Equipped with Corollary~\ref{thm:combination} we are ready to approach
the proof of Theorem~\ref{thm_aa}.

\begin{proof}[Proof of Theorem~\ref{thm_aa}]
  One starts from the Floquet Hamiltonian $K=D+H+V(t)$ and applies to
  it $q$ times Corollary~\ref{thm:combination}, with the steps being
  enumerated by $i=1,2,\ldots,q$. In the $i$th step one assumes that a
  strongly continuous function $J_{i-1}(t)$ with values in unitary
  operators on $\Hs$ has been already constructed so that
  \begin{displaymath}
    K = J_{i-1}(t)\left(D+H+A_{i-1}+B_{i-1}(t)\right)J_{i-1}(t)^*,
  \end{displaymath}
  with $A_{i-1}\in\YY(\infty,\gamma)$ being Hermitian and commuting
  with $H$, and
  \hspace{1em}
  $B_{i-1}(t)\in\linebreak
  \YY(p-i+1,i\gamma)$
  being symmetric, $T$-periodic and $C^1$ in the strong sense. In the
  first step one sets $A_0:=0$, $B_0(t):=V(t)$ and $J_0(t):=1$.
  
  Corollary~\ref{thm:combination} can be applied to the Floquet
  Hamiltonian $K_{i-1}:=D+H+A_{i-1}+B_{i-1}(t)$, with $r=p-i+1$,
  provided there is satisfied the assumption
  \begin{equation}
    \label{A_i}
    \Vert A_{i-1}\Vert_{\infty,\gamma}
    +\Vert B_{i-1}\Vert_{p-i+1,i\gamma}
    \leq \frac{c_{H}}{4\pi C_{p-i+2}}.
  \end{equation}
  Recall that the constant $C_p$ is given by (\ref{norm_const}). Under
  this assumption, there exists a strongly differentiable family of
  unitary operators $\UU_i(t)$ such that
  \begin{displaymath}
    K_{i} := D+H+A_{i}+B_{i}(t)
    = \UU_{i}(t) K_{i-1}\,\UU_{i}(t)^*
  \end{displaymath}
  where $A_{i}\in\YY(\infty,\gamma)$ is symmetric and diagonal, and
  $B_{i}(t)\in\YY(p-i,(i+1)\gamma)$ is $T$-periodic, symmetric and
  strongly $C^1$. Moreover,
  \begin{equation}
    \label{novy_diag}
    \Vert A_{i}\Vert_{\infty,\gamma}
    \leq 2\left(\Vert A_{i-1}\Vert _{\infty,\gamma}
      +\Vert B_{i-1}\Vert_{p-i+1,i\gamma}\right)
  \end{equation}
  and
  \begin{eqnarray}
    \nonumber
    \Vert B_i\Vert_{p-i,(i+1)\gamma}
    &\leq& \frac{1}{2C_{p-i+1}}
    \exp\!\left(2\,\frac{C_{p-i+1}}{C_{p-i+2}}\right)
    \big(\exp\!\left(4C_{p-i+1}T\,
      \Vert B_{i-1}\Vert_{p-i+1,i\gamma}\right)-1\big) \\
    \label{norm A_i}
    && \textrm{~}\times\,
    \left(C_H+4\Vert A_{i-1}\Vert_{\infty,\gamma}
      +2C_{p-i+1}\Vert B_{i-1}\Vert_{p-i+1,i\gamma}\right).
  \end{eqnarray}
  Finally, $J_{i}(t):=J_{i-1}(t)\UU_i(t)^*$ is a family of unitary
  operators which is continuously differentiable in the strong sense
  and such that
  \begin{displaymath}
    K = J_{i}(t)\left(D+H+A_{i}+B_{i}(t)\right)J_{i}(t)^*.
  \end{displaymath}

  To finish the proof we have to choose $\varepsilon>0$ sufficiently
  small so that if $\|V\|_{p,\gamma}<\varepsilon$ then condition
  (\ref{A_i}) is satisfied in each step $i=1,2,\ldots,q$.

  From (\ref{novy_diag}) one derives by induction
  \begin{displaymath}
    \Vert A_{i} \Vert_{\infty,\gamma}
    \leq \sum_{j=0}^{i-1} 2^{i-j}\Vert B_j\Vert_{p-j,(j+1)\gamma}.
  \end{displaymath}
  From here we deduce that inequalities (\ref{A_i}) are satisfied for
  $i=1,2,\ldots,k$, provided the inequalities
  \begin{equation}
    \label{B_i}
    \sum_{j=0}^{i-1} 2^{i-1-j} \Vert B_j\Vert_{p-j,(j+1)\gamma}
    \leq \frac{c_{H}}{4\pi C_{p-i+2}}
  \end{equation}
  are satisfied for the same range of indices. Furthermore, relations
  (\ref{A_i}) and (\ref{norm A_i}) imply that
  \begin{equation}
    \label{eq:B_leq_phi}
    \|B_i\|_{p-i,(i+1)\gamma} \leq \phi_i(\|B_{i-1}\|_{p-i+1,i\gamma})
  \end{equation}
  where
  \begin{displaymath}
    \phi_i(y) :=
    \frac{\exp\!\left(2\,\frac{C_{p-i+1}}{C_{p-i+2}}\right)}{2C_{p-i+1}}
    \big(\exp(4C_{p-i+1}T\,y)-1\big)
    \left(C_H+\frac{c_{H}}{\pi C_{p-i+2}}+(2C_{p-i+1}-4)y\right).
  \end{displaymath}
  Set
  \begin{displaymath}
    F_i(y) := 2^{i-1}y+\sum_{j=1}^{i-1} 2^{i-1-j}
    \phi_{j}\circ\phi_{j-1}\circ\cdots\circ\phi_1(y),\textrm{~~}
    i=1,2,\ldots,q.
  \end{displaymath}
  It follows from (\ref{eq:B_leq_phi}) that inequalities (\ref{B_i})
  are satisfied for $i=1,2,\ldots,k$, if it holds
  \begin{displaymath}
    F_i(\|B_0\|_{p,\gamma}) \leq \frac{c_{H}}{4\pi C_{p-i+2}}
  \end{displaymath}
  for the same range of indices.

  Recall that $B_0(t)=V(t)$. From this discussion it is clear that
  condition (\ref{A_i}) is satisfied in all steps $i=1,2,\ldots,q$,
  provided $\|V\|_{p,\gamma}\leq\varepsilon$ and $\varepsilon>0$ is
  chosen so that
  \begin{displaymath}
    \forall i\in\{1,2,\ldots,q\},\forall y\in\,[0,\varepsilon\,],
    \textrm{~}F_i(y) \leq \frac{c_{H}}{4\pi C_{p-i+2}}\,.
  \end{displaymath}
  But all functions $\phi_i(y)$ are continuous, strictly increasing
  and satisfy $\phi_i(0)=0$. Consequently, the same is true for all
  functions $F_i(y)$. Hence the following choice of $\varepsilon$ will
  do:
  \begin{displaymath}
    \varepsilon = \min\left\{
      F_i^{\,-1}\!\left(\frac{c_{H}}{4\pi C_{p-i+2}}\right);
    \textrm{~}1\leq i\leq q\right\}.
  \end{displaymath}
  This completes the
  proof of Theorem~\ref{thm_aa}.
\end{proof}

\section*{Acknowledgments}
The authors wish to acknowledge gratefully partial support from the
following grants: grant No.~201/05/0857 of the Grant Agency of the
Czech Republic (P.~\v{S}.), grant No.~MSM 6840770010 of the Ministry
of Education of the Czech Republic (O.~L.), and grant No.~LC06002 of
the Ministry of Education of the Czech Republic (O.~L. and P.~\v{S}.).


\begin{thebibliography}{10} 
  
\bibitem{ADE} Asch J., Duclos P., Exner P., \emph{Stability of driven
    systems with growing gaps, quantum rings, and Wannier ladders},
  J.~Stat. Phys. \textbf{92} (1998) 1053-1070.

\bibitem{ABCF} Astaburuaga M.~A., Bourget O., Cortés V.~H., Fernández
  C., \emph{Floquet operators without singular continuous spectrum}
  J.~Funct. Anal. \textbf{238} (2006) 489-517.

\bibitem{BJ} Barbaroux J.~M., Joye A., \emph{Expectation values of
    observables in time-dependent quantum mechanics}, J.~Stat. Phys.
  \textbf{90} (1998) 1225-1249.

\bibitem{BJLPN} Bunimovich L., Jauslin H.~R., Lebowitz J.~L.,
  Pellegrinotti A., Nielaba P., \emph{Diffusive energy growth in
    classical and quantum driven oscillators}, J.~Stat. Phys.
  \textbf{62} (1991) 793-817.

\bibitem{DeBF} De~Bi\`evre S., Forni G., \emph{Transport properties of
    kicked and quasiperiodic Hamiltonians}, J.~Stat. Phys. \textbf{90}
  (1998) 1201-1223.

\bibitem{BR} Bhatia R., Rosenthal P., \emph{How and why to solve the
    operator equation $A X- X B=Y$}, Bull. London Math. Soc.
  \textbf{29} (1997) 1-21.

\bibitem{Bou} Bourget O., \emph{Singular continuous Floquet operators
    for systems with increasing gaps}, J.~Math. Anal. Appl.
  \textbf{276} (2002) 28-39.

\bibitem{Bou2} Bourget O., \emph{Singular continuous Floquet operator
    for periodic quantum systems}, J.~Math. Anal. Appl. \textbf{301}
  (2005) 65-83.

\bibitem{Cb} Combes J.-M., \emph{Connection between quantum dynamics
    and spectral properties of time-evolution operators}, in
  ``Differential Equations and Applications to Mathematical physics'',
  W.~F. Ames, E.~M. Harrell and J.~V. Herod eds. (Academic Press,
  Boston, 1993) pp. 59-68.

\bibitem{C1} Combescure M., \emph{The quantum stability problem for
    time-periodic perturbations of the harmonic oscillator}, Ann.
  Inst. Henri Poincar\'e \textbf{47} (1987) 62-82, Erratum: Ann.
  Inst. Henri Poincar\'e 47 (1987) 451-454.

\bibitem{C2} Combescure M., \emph{Spectral properties of a
    periodically kicked quantum Hamiltonian}, J.~Stat. Phys.
  \textbf{59} (1990) 679-690.

\bibitem{C3} Combescure M., \emph{Recurrent versus diffusive dynamics
    for a kicked quantum oscillator}, Ann. Inst.  H.~Poincar\'e
  \textbf{57} (1992) 67-87.

\bibitem{DLSV} Duclos~P., Lev O., \v{S}\v{t}ov\'\i\v{c}ek~P.,
  Vittot~M.: \emph{Progressive diagonalization and applications}, in
  "Operator Algebras and Mathematical Physics", J.-M. Combes et al.
  eds. (The Theta Foundation, Bucharest, 2003) pp. 75-88.

\bibitem{DSSV} Duclos~P., Soccorsi~E., \v{S}\v{t}ov\'\i\v{c}ek~P.,
  Vittot~M.: \emph{Dynamical localization in periodically driven
    quantum systems}, in "Advances in Operator Algebras and
  Mathematical Physics", F.-P.~Boca, O.~Bratteli, R.~Longo and
  H.~Siedentop eds. (The Theta Foundation, Bucharest, 2005) pp.
  57-66.

\bibitem{EV} Enss V., Veseli\'c K., \emph{Bound states and propagating
    states for time-dependent Hamiltonians}, Ann. Inst.  H.~Poincar\'e
  \textbf{39} (1983) 159-191.

\bibitem{GM} Guarneri I., Mantica G., \emph{On the asymptotic
    properties of quantum dynamics in the presence of a fractal
    spectrum}, Ann. Inst. H.~Poincar\'e~A \textbf{61} (1994) 369-379.

\bibitem{HLS} Hagedorn G.~A., Loss M., Slawny J.,
  \emph{Non-stochasticity of time-dependent quadratic Hamiltonians and
    the spectra of canonical transformations}, J.~Phys.~A: Math. Gen.
  \textbf{19} (1986) 521-531.

\bibitem{H1} Howland J.~S., \emph{Floquet operators with singular
    spectrum, I}, Ann. Inst. H.~Poincar\'e, \textbf{50} (1989)
  309-323.

\bibitem{H2} Howland J.~S., \emph{Floquet operators with singular
    spectrum, II}, Ann. Inst. H.~Poincar\'e, \textbf{50} (1989)
  325-334.

\bibitem{H3} Howland J.~S., \emph{Floquet operators with singular
    spectrum, III}, Ann. Inst. H.~Poincar\'e, \textbf{69} (1998)
  265-273.

\bibitem{JL} Jauslin H., Lebowitz J.~L., \emph{Spectral and stability
    aspects of quantum chaos }, Chaos, \textbf{1} (1991) 114-121.

\bibitem{J1} Joye A., \emph{Absence of absolutely continuous spectrum
    of Floquet operators}, J.~Stat. Phys. \textbf{75} (1994) 929-952.

\bibitem{J2} Joye A., \emph{Upper bounds for the energy expectation in
    the time-dependent quantum mechanics}, J.~Stat. Phys. \textbf{85}
  (1996) 575-606.

\bibitem{K} Krein S.~G., \emph{Linear Differential Equations in Banach
    Spaces} (AMS, Providence, Rhode Island, 1971)

\bibitem{McCMcK} McCaw J., McKellar B., \emph{On the continuous
    spectral component of the Floquet operator for a periodically
    kicked quantum system} J.~Math. Phys. \textbf{46} (2005) 103503.

\bibitem{N1} Nenciu G., \emph{Floquet operators without absolutely
    continuous spectrum}, Ann. Inst. H.~Poincar\'e~A \textbf{59}
  (1993) 91-97.

\bibitem{N2} Nenciu G., \emph{Adiabatic theory: Stability of systems
    with increasing gaps}, Ann. Inst. H.~Poincar\'e~A \textbf{67}
  (1997) 411-424.

\bibitem{dO} de~Oliveira C.~R., \emph{Some remarks concerning
    stability for nonstationary quantum systems}, J.~Stat. Phys.
  \textbf{78} (1995) 1055-1065.

\bibitem{dOS} de~Oliveira C.~R., Simsen M.~S., \emph{A Floquet
    operator with pure point spectrum and energy instability}, Ann.
  Inst. H.~Poincar\'e, to appear.

\bibitem{RS} Reed M., Simon B., \emph{Methods of Modern Mathematical
    Physics III}, Academic Press, San Diego, 1979.



\end{thebibliography}
\end{document}